%% file: main.tex
\RequirePackage{luatex85}

\documentclass[11pt,a4paper]{article}
\usepackage{iftex}

\ifPDFTeX
\usepackage[utf8]{inputenc}
\usepackage[T1]{fontenc}
\fi

\usepackage{amsmath}
\usepackage{amssymb}
\usepackage{amsthm}
\usepackage{thmtools}

\ifPDFTeX

\usepackage{libertine}
\usepackage[libertine]{newtxmath} 
\usepackage[scaled=0.96]{zi4} 
\else
\usepackage{fontspec}
\usepackage{unicode-math}
\fi

\usepackage[DIV=11]{typearea} 

\usepackage{microtype}

\usepackage{mathtools}

\usepackage[procnumbered,ruled,vlined,linesnumbered,algo2e]{algorithm2e}

\usepackage{xcolor}

\usepackage{comment}

\usepackage[
	style=alphabetic,
	backref=true,
	doi=false,
	url=false,
	maxcitenames=3,
	mincitenames=3,
	maxbibnames=10,
	minbibnames=10,
	backend=biber,
	sortlocale=en_US
]{biblatex}

\usepackage[ocgcolorlinks]{hyperref} 
\usepackage{cleveref}

\input{commands}

\colorlet{DarkRed}{red!50!black}
\colorlet{DarkGreen}{green!50!black}
\colorlet{DarkBlue}{blue!50!black}

\hypersetup{
	linkcolor = DarkRed,
	citecolor = DarkGreen,
	urlcolor = DarkBlue,
	bookmarksnumbered = true,
	linktocpage = true
}

\declaretheorem[numberwithin=section]{theorem}
\declaretheorem[numberlike=theorem]{lemma}

\declaretheorem[numberlike=theorem]{corollary}
\declaretheorem[numberlike=theorem]{definition}
\declaretheorem[numberlike=theorem]{claim}

\declaretheorem[numberlike=theorem]{observation}
\declaretheorem[numberlike=theorem, style=remark]{remark}
\crefname{algorithm}{Procedure}{Procedures}
\Crefname{algorithm}{Procedure}{Procedures}

\allowdisplaybreaks[1]

\DontPrintSemicolon
\SetKw{KwAnd}{and}
\SetProcNameSty{textsc}
\SetFuncSty{textsc}

\DeclarePairedDelimiterX{\set}[2]{\{}{\}}{#1\,\delimsize|\,\mathopen{}#2}

\DeclarePairedDelimiterX{\floor}[1]{\lfloor}{\rfloor}{#1}
\newcommand{\fraction}{\operatorname{frac}}
\newcommand{\nil}{\textsc{nil}}
\newcommand{\expdis}{\operatorname{Exp}}
\newcommand{\str}{\operatorname{stretch}}
\newcommand{\avestr}{\operatorname{avg-stretch}}
\newcommand{\capacity}{\operatorname{cap}}
\newcommand{\poly}{\operatorname{poly}}
\newcommand{\polylog}{\operatorname{polylog}}
\newcommand{\dist}{\textrm{dist}}

\newcommand{\lev}{\ell}

\newcommand{\degree}{\operatorname{deg}}
\newcommand{\fract}{\beta}
\newcommand{\diam}{\Delta}

\SetKwProg{Procedure}{Procedure}{}{}
\SetKwFunction{Insert}{insert}
\SetKwFunction{Delete}{delete}
\SetKwFunction{Relax}{relax}
\SetKwFunction{Query}{query}
\SetKwFunction{Distance}{distance}
\SetKwFunction{Partition}{partition}
\SetKwFunction{Initialize}{initialize}
\SetKwFunction{Increase}{increase}
\SetKwFunction{UpdateLevels}{updateLevels}

\title{Dynamic Low-Stretch Trees via Dynamic Low-Diameter~Decompositions\thanks{To be presented at the 51st Annual ACM Symposium on the Theory of Computing (STOC 2019)}}

\author{
Sebastian Forster\thanks{University of Salzburg, Department of Computer Sciences, Austria. Work partially done while at University of Vienna.}
\and
Gramoz Goranci\thanks{University of Vienna, Faculty of Computer Science, Austria. The research leading to these results has received funding from the European Research Council under the European Union's Seventh Framework Programme (FP/2007-2013) / ERC Grant Agreement no. 340506. This work was done in part while visiting Georgia Institute of Technology.}
}

\date{}

\hypersetup{
	pdftitle = {Dynamic Low-Stretch Trees via Dynamic Low-Diameter Decompositions},
	pdfauthor = {Gramoz Goranci, Sebastian Krinninger}
}

\begin{document}
\maketitle
\begin{abstract}
\input{abstract}
\end{abstract}

\input{introduction}
\input{prelim}
\input{overview}
\input{low_stretch_tree}
\input{LDD}
\input{spanner}

\subsection*{Acknowledgements}
We thank Pavel Kolev and Thatchaphol Saranurak for inspiring discussions.
We are indebted to Richard Peng for pointing out some significant slack in our initial construction.
We thank Tianyi Zhang for finding a bug in a previous version of this paper.


\printbibliography[heading=bibintoc] 

\end{document}

%% file: commands.tex
\makeatletter
\g@addto@macro\bfseries{\boldmath}
\makeatother

\makeatletter
\g@addto@macro\mdseries{\unboldmath}
\g@addto@macro\normalfont{\unboldmath}
\g@addto@macro\rmfamily{\unboldmath}
\g@addto@macro\upshape{\unboldmath}
\makeatother

\DeclareCiteCommand{\citem}
    {}
    {\mkbibbrackets{\bibhyperref{\usebibmacro{postnote}}}}
    {\multicitedelim}
    {}

\renewcommand*{\multicitedelim}{\addcomma\space}

\AtEveryCitekey{\clearfield{note}}

\newcommand{\myhref}[1]{%
  \iffieldundef{doi}
    {\iffieldundef{url}
       {#1}
       {\href{\strfield{url}}{#1}}}
    {\href{http://dx.doi.org/\strfield{doi}}{#1}}%
}

\DeclareFieldFormat{title}{\myhref{\mkbibemph{#1}}}
\DeclareFieldFormat
  [article,inbook,incollection,inproceedings,patent,thesis,unpublished]
  {title}{\myhref{\mkbibquote{#1\isdot}}}

\makeatletter
\AtBeginDocument{%
    \newlength{\temp@x}%
    \newlength{\temp@y}%
    \newlength{\temp@w}%
    \newlength{\temp@h}%
    \def\my@coords#1#2#3#4{%
      \setlength{\temp@x}{#1}%
      \setlength{\temp@y}{#2}%
      \setlength{\temp@w}{#3}%
      \setlength{\temp@h}{#4}%
      \adjustlengths{}%
      \my@pdfliteral{\strip@pt\temp@x\space\strip@pt\temp@y\space\strip@pt\temp@w\space\strip@pt\temp@h\space re}}%
    \ifpdf
      \typeout{In PDF mode}%
      \def\my@pdfliteral#1{\pdfliteral page{#1}}
      \def\adjustlengths{}%
    \fi
    \ifxetex
      \def\my@pdfliteral #1{}
      \def\adjustlengths{\setlength{\temp@h}{-\temp@h}\addtolength{\temp@y}{1in}\addtolength{\temp@x}{-1in}}%
    \fi%
    \def\Hy@colorlink#1{%
      \begingroup
        \ifHy@ocgcolorlinks
          \def\Hy@ocgcolor{#1}%
          \my@pdfliteral{q}%
          \my@pdfliteral{7 Tr}
        \else
          \HyColor@UseColor#1%
        \fi
    }%
    \def\Hy@endcolorlink{%
      \ifHy@ocgcolorlinks%
        \my@pdfliteral{/OC/OCPrint BDC}%
        \my@coords{0pt}{0pt}{\pdfpagewidth}{\pdfpageheight}%
        \my@pdfliteral{F}
        \my@pdfliteral{EMC/OC/OCView BDC}%
        \begingroup%
          \expandafter\HyColor@UseColor\Hy@ocgcolor%
          \my@coords{0pt}{0pt}{\pdfpagewidth}{\pdfpageheight}%
          \my@pdfliteral{F}
        \endgroup%
        \my@pdfliteral{EMC}%
        \my@pdfliteral{0 Tr}
        \my@pdfliteral{Q}%
      \fi
      \endgroup
    }%
}
\makeatother

%% file: abstract.tex
Spanning trees of low average stretch on the non-tree edges, as introduced by Alon et al.~\citem[SICOMP 1995]{AlonKPW95}, are a natural graph-theoretic object. In recent years, they have found significant applications in solvers for symmetric diagonally dominant (SDD) linear systems. In this work, we provide the first dynamic algorithm for maintaining such trees under edge insertions and deletions to the input graph. Our algorithm has update time $ n^{1/2 + o(1)} $ and the average stretch of the maintained tree is $ n^{o(1)} $, which matches the stretch in the seminal result of Alon et al.

Similar to Alon et al., our dynamic low-stretch tree algorithm employs a dynamic hierarchy of low-diameter decompositions (LDDs).
As a major building block we use a dynamic LDD that we obtain by adapting the random-shift clustering of Miller et al.~\citem[SPAA 2013]{MillerPX13} to the dynamic setting.
%
The major technical challenge in our approach is to control the propagation of updates within our hierarchy of LDDs: each update to one level of the hierarchy could potentially induce several insertions \emph{and} deletions to the next level of the hierarchy.
We achieve this goal by a sophisticated amortization approach.
In particular, we give a bound on the number of changes made to the LDD per update to the input graph that is significantly better than the trivial bound implied by the update time.

We believe that the dynamic random-shift clustering might be useful for independent applications.
One of these applications is the dynamic spanner problem.
By combining the random-shift clustering with the recent spanner construction of Elkin and Neiman~\citem[SODA 2017]{ElkinN17}.
We obtain a fully dynamic algorithm for maintaining a spanner of stretch $ 2k - 1 $ and size $ O (n^{1 + 1/k} \log{n}) $ with amortized update time $ O (k \log^2 n) $ for any integer $ 2 \leq k \leq \log n $.
Compared to the state-of-the art in this regime \citem[Baswana et al.~TALG '12]{BaswanaKS12}, we improve upon the size of the spanner and the update time by a factor of $ k $.

%% file: introduction.tex
\section{Introduction}

Graph compression is an important paradigm in modern algorithm design.
Given a graph $ G $ with $ n $ nodes, can we find a substantially smaller (read: sparser) subgraph $ H $ such that $ H $ preserves central properties of~$ G $?
Very often, this compression is ``lossy'' in the sense that the properties of interest are only preserved approximately.
A ubiquitous example of graph compression schemes are \emph{spanners}:
every graph $ G $ admits a spanner $ H $ with $ O (n^{1 + 1/k}) $ edges that has \emph{stretch}~$ 2 k - 1 $ (for any integer $ k \geq 2 $), meaning that for every edge $ e = (u, v) $ of~$ G $ not present in $ H $ there is a path from $ u $ to $ v $ in~$ H $ of length at most $ 2 k - 1 $.
Thus, when $ k = \log{n} $, very succinct compression with $ O (n) $ edges can be achieved at the price of stretch $ O (\log n) $.

The most succinct form of subgraph compression is achieved when $ H $ is a tree.
Spanning trees, for example, are a well-known tool for preserving the connectivity of a graph.
It is thus natural to ask whether, similar to spanners, one could also have spanning trees with low stretch for each edge.
This unfortunately is known to be false: in a ring of $ n $ nodes every tree will result in a stretch of $ n - 1 $ for the single edge not contained in the tree.
However, it turns out that a quite similar goal can be achieved by relaxing the concept of stretch:
every graph $ G $ admits a spanning tree $ T $ of \emph{average stretch} $ O (\log{n} \log \log n) $~\cite{AbrahamN12}, where the average stretch is the sum of the stretches of all edges divided by the total number of edges.
Such subgraphs are called \emph{low (average) stretch trees} and have found numerous applications in recent years, most notably in the design of fast solvers for symmetric diagonally dominant (SDD) linear systems~\cite{SpielmanT14,KoutisMP14,BlellochGKMPT14,KoutisMP11,KelnerOSZ13,CohenKMPPRX14}.
We believe that their fundamental graph-theoretic motivation and their powerful applications make low-stretch trees a very natural object to study as well in a dynamic setting, similar to spanners~\cite{AusielloFI06,Elkin11,BaswanaKS12,BodwinK16} and minimum spanning trees~\cite{Frederickson85,EppsteinGIN97,HenzingerK01,HolmLT01,Wulff-Nilsen17,NanongkaiSW17}.
Indeed, the design of a dynamic algorithm for maintaining a low-stretch tree was posed as an open problem by Baswana et al.~\cite{BaswanaKS12}, but despite extensive research on dynamic algorithms in recent years, no such algorithm has yet been found.

In this paper, we give the first non-trivial algorithm for this problem in the \emph{dynamic} setting.
Specifically, we maintain a low-stretch tree $ T $ of a dynamic graph~$ G $ undergoing updates in the form of edge insertions and deletions in the sense that after each update to $ G $ we compute the set of necessary changes to $ T $.
The goal in this problem is to keep the time spent after each update small while still keeping the average stretch of $ T $ tolerable.
Our main result is a fully dynamic algorithm for maintaining a spanning tree of expected average stretch~$ n^{o(1)} $ with expected amortized update time~$ n^{1/2 + o(1)} $.
At a high level, we obtain this result by combining the classic low-stretch tree construction of Alon et al.~\cite{AlonKPW95} with a dynamic algorithm for maintaining low-diameter decompositions (LDDs) based on random-shift clustering~\cite{MillerPX13}.
Our LDD algorithm might be of independent interest, and we provide another application by using it to obtain a dynamic version of the recent spanner construction of Elkin and Neiman~\cite{ElkinN17}.
The resulting dynamic spanner algorithm improves upon one of the state-of-the-art algorithms by Baswana et al.~\cite{BaswanaKS12}.

Our overall approach towards the low-stretch tree algorithm -- to use low-diameter decompositions based on random-shift clustering in the construction of Alon et al.~\cite{AlonKPW95} -- has been used before in parallel and distributed algorithms~\cite{BlellochGKMPT14,GhaffariKKLP15,HaeuplerL18}.
However, to make this approach work in the dynamic setting we need to circumvent some non-trivial challenges.
In particular, we cannot employ the following paradigm that often is very helpful in designing dynamic algorithms: design an algorithm that can only handle edge deletions and then extend it to the fully dynamic setting using a general reduction.
While we do follow this paradigm for our dynamic LDD algorithm, there are two obstacles that prevent us from doing so for the dynamic low-stretch tree:
First, many fully-dynamic-to-decremental reductions exploit some form of ``decomposability'', which does not hold for low-stretch trees, i.e., low-stretch trees of subgraphs of the input graph cannot be simply be combined to a single low-stretch tree of the full graph.
Second, in our dynamic low-diameter decomposition edges might start and stop being inter-cluster edges, even if the input graph is only undergoing deletions.
In the hierarchy of Alon et al.\ this leads to both insertions and deletions at the next level of the hierarchy.
As opposed to other dynamic problems~\cite{HenzingerKN16,AbrahamDKKP16}, one cannot simply enforce some type of ``monotonicity'' by not passing on insertions to the next level of the hierarchy (to stay within a deletions-only setting) as there might be too many such edges to ignore them. 
Thus, it seems that we really have to deal with the fully dynamic setting in the first place.
We show that this can be done by a sophisticated amortization approach that explicitly analyzes the number of updates passed on to the next level.

\paragraph*{Related Work.}

Low-stretch trees have been introduced by Alon et al.~\cite{AlonKPW95} who obtained an average stretch of $ 2^{O (\sqrt{\log n \log \log n})} $ and also gave a lower bound of $ \Omega (\log n) $ on the average stretch.
The first construction with polylogarithmic average stretch was given by Elkin et al.~\cite{ElkinEST08}.
Further improvements~\cite{AbrahamBN08,KoutisMP11} culminated in the state-of-the-art construction of Abraham and Neiman~\cite{AbrahamN12} with average stretch $ O(\log n \log \log n) $.
All these trees with polylogarithmic average stretch can be computed in time $ \tilde O (m) $.\footnote{Throughout this paper we will use $ \tilde O (\cdot) $-notation to suppress factors that are polylogarithmic in $ n $.}
To the best of our knowledge, all known algorithms in parallel and distributed models of computation~\cite{BlellochGKMPT14,GhaffariKKLP15,HaeuplerL18} are based on the scheme of Alon et al.\ and thus do not provide polylogarithmic stretch guarantees.

The main application of low-stretch trees has been in solving symmetric, diagonally dominant (SDD) systems of linear equations.
It has been observed that iterative methods for solving these systems can be made faster by preconditioning with a low-stretch tree for weighted graphs~\cite{Vaidya91,BomanH03,SpielmanW09}.
Consequently, they have been an important ingredient in the breakthrough result of Spielman and Teng~\cite{SpielmanT14} for solving SDD systems in nearly linear time.
In this solver, low-stretch trees are utilized for constructing ultra-sparsifiers, which in turn are used as preconditioners.
Beyond this initial breakthrough, low-stretch trees have also been used in subsequent, faster solvers~\cite{KoutisMP14,BlellochGKMPT14,KoutisMP11,KelnerOSZ13,CohenKMPPRX14}.
Another prominent application of low-stretch trees~(concretely, the variant of random spanning trees with low expected stretch) is the remarkable cut-based graph decomposition of R\"acke~\Cite{Racke08,AndersenF09}, which embeds any general undirected graph into a convex combination of spanning trees, while paying a congestion of only $\tilde{O}(\log n)$ for the embedding. This decomposition tool, initially aimed at giving the best competitive ratio for oblivious routing, has found several applications ranging from approximation algorithms for cut-based problems~(e.g., minimum bisection~\cite{Racke08}) to graph compression~(e.g., vertex sparsifiers~\cite{Moitra09}). 
Other classic problems in the realm of approximation algorithms that utilize the properties of low-stretch trees include the $ k $-server problem~\cite{AlonKPW95} and the minimum communication cost spanning tree problem~\cite{Hu74,PelegR98}.

In terms of dynamic algorithms, we are not aware of any prior work for maintaining low-stretch trees.
The closest related works are arguably dynamic algorithms for maintaining distance oracles and spanners, as they also aim preserving pairwise distances, and dynamic algorithms for maintaining minimum spanning trees, as they also are spanning trees with an additional property.


A distance oracle is a data structure that can answer queries for the (approximate) distance between a pair of nodes.
The fully dynamic distance oracle of Abraham, Chechik, and Talwar~\cite{AbrahamCT14} for unweighted, undirected graphs has expected amortized update time $ \tilde O (\sqrt{m} n^{1/k}) $, query time $ O (k^2 \rho^2) $, and stretch $ 2^{O (k \rho)} $, where the parameter $ k \geq 2 $ is integer and $ \rho = 1 + \lceil \tfrac{\log n^{1 - 1/k}}{\log(m / n^{1 - 1/k})} \rceil $.
To the best of our knowledge, the recent decremental distance oracle of Chechik~\cite{Chechik18} can be used to extend this result to weighted graphs and to improve the stretch and the query time, while leaving the update time essentially unchanged.

For dynamic spanner algorithms, the main goal is to maintain, for any given integer $ k \geq 2 $, a spanner of stretch $ 2k - 1 $ with $ \tilde O (n^{1 + 1/k}) $ edges.
Spanners of stretch $ 2k - 1 $ and size $ O (n^{1 + 1/k}) $ exist for every undirected graph~\cite{Awerbuch85}, and this trade-off is presumably tight under Erd\H{o}s's girth conjecture.
The dynamic spanner problem has been introduced by Ausiello et al.~\cite{AusielloFI06}.
They showed how to maintain a $3$- or $5$-spanner with amortized update time $ O(\Delta) $, where $ \Delta \leq n $ is the maximum degree of the graph.
Using techniques from the streaming literature, Elkin~\cite{Elkin11} provided an algorithm for maintaining a $ (2k - 1) $-spanner with $ \tilde O (m n^{-1/k}) $ expected update time.
Faster update times were achieved by Baswana et al.~\cite{BaswanaKS12}: their algorithms maintain $ (2k - 1) $-spanners either with expected amortized update time $ O(1)^k $ or with expected amortized update time $ O (k^2 \log^2 n) $.
Later, Bodwin and Krinninger~\cite{BodwinK16} initiated the study of dynamic spanners with worst-case update times, and recently, Bernstein, Forster, and Henzinger~\cite{BernsteinFH19} presented a deamortization approach to maintain $ (2k - 1) $-spanners with high-probability worst-case update time $ O(1)^k \log^3 n $.
All of these algorithms exhibit the stretch/space trade-off mentioned above in unweighted graphs, up to polylogarithmic factors in the size of the spanner.

The first non-trivial algorithm for dynamically maintaining a minimum spanning tree was developed by Frederickson~\cite{Frederickson85} and had a worst-case update time of $ O (\sqrt{m}) $.
Using a general sparsification technique, this bound was improved to $ O (\sqrt{n}) $ by Eppstein et al.~\cite{EppsteinGIN97}.
In terms of amortized bounds, Holm et al.~\cite{HolmLT01} were the first to improve this bound and obtained polylogarithmic amortized update time.
In a recent breakthrough, Nanongkai, Saranurak, and Wulff-Nilsen~\cite{Wulff-Nilsen17,NanongkaiS17,NanongkaiSW17} finally achieved a \emph{worst-case} update time of $ n^{o(1)} $.

\paragraph*{Our Results.}

Our main result is a dynamic algorithm for maintaining a low-stretch tree of an unweighted, undirected graph.

\begin{theorem}\label{thm:fully dynamic low stretch tree}
Given any unweighted, undirected graph undergoing edge insertions and deletions, there is a fully dynamic algorithm for maintaining a spanning forest of expected average stretch $ n^{o(1)} $ that has expected amortized update time $ m^{1/2 + o(1)} $.
These guarantees hold against an oblivious adversary.
\end{theorem}

This is the first non-trivial algorithm for this fundamental problem.
Our stretch matches the seminal construction of Alon et al.~\cite{AlonKPW95}, which is still the state of the art in parallel and distributed settings~\cite{BlellochGKMPT14,GhaffariKKLP15,HaeuplerL18}.

Similar to the approach of~\cite{KoutisLP16} in the \emph{static} setting, we can apply Theorem~\ref{thm:fully dynamic low stretch tree} to a cut sparsifier of the input graph, which has only $ \tilde O (n) $ edges, to improve the running time for dense graphs.
Such a cut sparsifier can be maintained with the dynamic algorithm of Abraham et al.~\cite{AbrahamDKKP16} that has polylogarithmic update time.

\begin{corollary}\label{cor:fully dynamic low stretch tree}
Given any unweighted, undirected graph undergoing edge insertions and deletions, there is a fully dynamic algorithm for maintaining a spanning forest of expected average stretch $ n^{o(1)} $ that has expected amortized update time $ n^{1/2 + o(1)} $.
These guarantees hold against an oblivious adversary.
\end{corollary}

Obtaining this improvement is non-trivial because cut sparsifiers are weighted graphs, even when the input graph is unweighted, and the algorithm of Theorem~\ref{thm:fully dynamic low stretch tree} only accepts unweighted graphs.
To deal with this issue, we deviate from the approach of~\cite{KoutisLP16} by interpreting the edge weights of the sparsifier as edge multiplicities in an unweighted graph.
A fine-grained analysis of the amount of change to edge the multiplicities per update to the input graph then allows us to get the desired benefits of combining both algorithms.


We additionally show that $ \sqrt{n} $ is not an inherent barrier to the update time, at least if very large stretch is tolerated.
A modification of our algorithm gives average stretch $ O(t) $ and update time $ \tfrac{n^{1 + o(1)}}{t} $ for $ t \geq \sqrt{n} $.

One of the main building blocks of our dynamic low-stretch tree algorithm is a dynamic algorithm for maintaining a low-diameter decomposition (LDD). Roughly speaking, for $\beta \in (0,1)$ and $\diam > 0$, a $(\beta, \diam)$-decomposition of a graph is a partitioning of its nodes into disjoint clusters such that (1) any two nodes belonging to the same cluster are at distance at most $\diam$, and (2) the number of edges whose endpoints belong to different clusters is bounded by $\beta m$. The following theorem gives a dynamic variant of such decompositions.
 
\begin{theorem}\label{thm:fully dynamic LDD}
Given any unweighted, undirected multigraph undergoing edge insertions and deletions, there is a fully dynamic algorithm for maintaining a $ (\fract, O (\tfrac{\log{n}}{\fract})) $-decomposition (with clusters of strong diameter $ O (\tfrac{\log{n}}{\fract}) $ and at most $ \beta m $ inter-cluster edges in expectation) that has expected amortized update time $ O (\log^2{n} / \fract^2) $.
A spanning tree of diameter $ O (\tfrac{\log{n}}{\fract}) $ for each cluster can be maintained in the same time bound.
The expected amortized number of edges to become inter-cluster edges after each update is $ O (\log^2{n} / \fract) $.
These guarantees hold against an oblivious adversary.
\end{theorem}
Our algorithm is based on the random-shift clustering of Miller at al.~\cite{MillerPX13}, with many tweaks to make it work in a dynamic setting.
In our analysis of the algorithm, we bound the amortized number of changes to the clustering per update by $ \tilde O (1 / \fract) $, which is significantly smaller than the naive bound of $ \tilde O (1 / \fract^2) $ implied by the update time.
This is particularly important for hierarchical approaches, such as in our dynamic low-stretch tree algorithm, because a small bound on the number of amortized changes helps in controlling the number of induced updates to be processed within the hierarchy.
Independently, Saranurak and Wang~\cite{SaranurakW19} obtained a fully dynamic LLD algorithm with nearly the same guarantees (up to polylogarithmic factors).\footnote{The low-diameter decomposition of Saranurak and Wang can be maintained against an adaptive online adversary. However, the low-diameter spanning trees of their clustering can only be maintained against an oblivious adversary. Therefore, plugging in their dynamic LDD algorithm into our dynamic low-stretch tree construction does not yield any improvement over our guarantees.}
We believe that our solution is arguably simpler than their expander pruning approach.

The dynamic random-shift clustering underlying our dynamic LDD is of independent interest.
A direct consequence demonstrating the usefulness of our dynamic random-shift clustering algorithm is the following new result for the dynamic spanner problem.
\begin{theorem}\label{thm:fully dynamic spanner}
Given any unweighted, undirected graph undergoing edge insertions and deletions, there is a fully dynamic algorithm for maintaining a spanner of stretch $ 2k - 1 $ and expected size $ O (n^{1 + 1/k} \log{n}) $ that has expected amortized update time $ O (k \log^2{n}) $.
These guarantees hold against an oblivious adversary.
\end{theorem}
Recall that the fully dynamic algorithm of Baswana et al.~\cite{BaswanaKS12} maintains a spanner of stretch $ 2k - 1 $ and expected size $ O (k n^{1 + 1/k} \log n) $ with expected amortized update time $ O (k^2 \log^2{n}) $.
Our new algorithm thus improves both the size and the update time by a factor of~$ k $.
This is particularly relevant because the stretch/size trade-off of $ 2k - 1 $ vs.\ $ O (n^{1 + 1/k}) $ is tight under the girth conjecture.
We thus exceed the conjectured optimal size by a factor of only $ \log{n} $ compared to the prior $ k \log{n} $, where $ k $ might be as large as $ \log n $.
When we restrict ourselves to the decremental setting, we do achieve size $ O (n^{1 + 1/k}) $ with expected amortized update time $ O (k \log{n}) $.
Again, this saves a factor of~$ k $ compared to Baswana et al.~\cite{BaswanaKS12}.
To obtain Theorem~\ref{thm:fully dynamic spanner}, we employ our dynamic random-shift clustering algorithm in the spanner construction of Elkin and Neiman~\cite{ElkinN17} and combine it with the dynamic spanner framework of Baswana et al.~\cite{BaswanaKS12}.

\paragraph*{Structure of this Paper.}

The remainder of this paper is structured as follows.
We first settle the notation and terminology in Section~\ref{sec:prelim}.
We then give a high-level overview of our results and techniques in Section~\ref{sec:overview}.
Finally, we provide all necessary details for our dynamic low-stretch tree (Section~\ref{sec:low_stretch_tree}), our dynamic low-diameter decomposition (Section~\ref{sec:LDD}), and our dynamic spanner algorithm (Section~\ref{sec:spanner}).

%% file: prelim.tex
\section{Preliminaries}\label{sec:prelim}

\paragraph*{Graphs.}

Let $G=(V,E,w_G)$ be an undirected weighted graph, where $n = |V|$, $m = |E|$ and $w_G: E \rightarrow \mathbb{R}_+$. If $w_G(e) = 1$ for all $e \in E$, then we say $G$ is an undirected unweighted graph.
If $ E $ is a multiset, i.e., every element of $E$ may have integer multiplicity greater than $ 1 $, then we call $ G $ a multigraph.
For a subset $C \subseteq V$ let $G[C]$ denote the subgraph of $G$ induced by  $C$. Throughout the paper we call $C \subset V$ a \emph{cluster}. For any positive integer $k$, a \emph{clustering} of $G$ is a partition of $V$ into disjoint subsets $C_1, C_2, \ldots, C_k$. We say that an edge is an \emph{intra-cluster} edges if both its endpoints belong to the same cluster $C_i$ for some $i$; otherwise, we say that an edge is an \emph{inter-cluster} edge.

For any $u,v \in V$ let $\dist_G(u,v)$ denote the length of a shortest path between $u$ and $v$ induced by the edge weights $w_G$ of the graph $G$. When $G$ is clear from the context, we will omit the subscript. The \emph{strong diameter} of a cluster $C \subset V$ is the maximum length of the shortest path between any two nodes in $G[C]$, i.e., $\max \set{\dist_{G[C]}(u,v)}{u,v \in C}$. In the following we define a low-diameter clustering of $G$. 

\begin{definition} Let $\beta \in (0,1)$ and $\diam > 0$. Given an undirected, unweighted graph $G=(V,E)$, a $(\fract,\diam)$-decomposition of $G$ is a partition of $V$ into disjoint subsets $C_1, C_2, \ldots,C_k$ (for some $ k \geq 1 $) such that:
\begin{enumerate}
\item The strong diameter of each $C_i$ is at most $\diam$.
\item The number of edges with endpoints belonging to different subsets is at most $\fract m$.
\end{enumerate}
\end{definition}
In the $(\fract,\diam)$-decompositions of the randomized dynamic algorithms in this paper, the bound in Condition~2 is in expectation.

Let $H = (V, F) $ be a subgraph of $G = (V, E, w_G)$. For any pair of nodes $u,v \in V$, we let $\dist_H(u,v)$ denote the length of a shortest path between $u$ and $v$ in $H$. We define the \emph{stretch} of an edge $(u,v) \in E$ with respect to $H$ to be 
\begin{equation*}
	\str_H(u,v) \coloneqq  \frac{\dist_H(u,v)}{w_G(u,v)} \, .
\end{equation*}
The stretch of $ H $ is defined as the maximum stretch of any of edge $(u,v) \in E$.
The \emph{average stretch} over all edges of $G$ with respect to $H$ is given by
\[
	\avestr_H(G) \coloneqq \frac{1}{|E|} \sum_{(u,v) \in E} \str_H(u,v).
\]

\paragraph*{Exponential Distribution.}

For a parameter $\lambda$, the probability density function of the \emph{exponential distribution} $\expdis(\lambda)$ is given by
\[
	f(x,\lambda) \coloneqq  \begin{cases} 
   \lambda e^{-\lambda x} & \text{if } x \geq 0 \\
   0       & \text{otherwise}.
  \end{cases}
\]
The mean of the exponential distribution is $1/\lambda$.

\paragraph*{Dynamic Algorithms.}
Consider a graph with $ n $ nodes undergoing updates in the form of edge insertions and edge deletions.
An \emph{incremental} algorithm is a dynamic algorithm that can only handle insertions, a \emph{decremental} algorithm can only handle deletions, and a \emph{fully dynamic} algorithm can handle both.
We follow the convention that a fully dynamic algorithm starts from an empty graph with $ n $ nodes.
The (maximum) running time spent by a dynamic algorithm for processing each update (before the next update arrives) is called \emph{update time}.
We say that a dynamic algorithm has \emph{(expected) amortized} update time $ u (n) $ if its total running time spent for processing a sequence of $ q $ updates is bounded by $ q \cdot u (n) $ (in expectation).
In this paper, we assume that the updates to the graph are performed by an \emph{oblivious adversary} who fixes the sequences of updates in advance, i.e., the adversary is not allowed to adapt its sequence of updates as the algorithm proceeds.
This is a standard assumption in dynamic graph algorithms\footnote{For example, all known randomized dynamic spanner algorithms~\cite{Elkin11,BaswanaKS12,BodwinK16,BernsteinFH19} work under this assumption.} and it in particular implies that for randomized dynamic algorithms the sequence of updates is independent from the random choices of the algorithm.

%% file: overview.tex
\section{Technical Overview}\label{sec:overview}

In the following, we provide some intuition for our approach and highlight the main ideas of this paper.

\paragraph*{Low-Stretch Tree.}
A first idea is to employ the dynamic low-diameter decomposition of Theorem~\ref{thm:fully dynamic LDD}.
This algorithm can maintain a $ (\fract, O(\tfrac{\log{n}}{\fract})) $-decomposition, i.e., a partitioning of the graph into clusters such that there are at most $ \fract m $ inter-cluster edges and the (strong) diameter of each cluster is at most $ O (\tfrac{\log{n}}{\fract}) $.
In particular, each cluster has a designated center and the algorithm maintains a spanning tree of each cluster in which every node is at distance at most $ O (\tfrac{\log{n}}{\fract}) $ from the center.
Now consider the following simple dynamic algorithm:
\begin{enumerate}
\item Maintain a $ (\fract, O(\tfrac{\log{n}}{\fract})) $-decomposition of the input graph~$ G $.
\item Contract the clusters in the decomposition to single nodes and maintain a multigraph $ G' $ containing one node for each cluster and all inter-cluster edges.
\item Compute a low-stretch tree $ T' $ of $ G' $ after each update to~$ G $ using a static algorithm providing polylogarithmic average stretch.
\item Maintain $ T $ as the ``expansion'' of $ T' $ in which every node in $ T' $ is replaced by the spanning tree of diameter $ O (\tfrac{\log{n}}{\fract}) $ of the cluster representing the node.
\end{enumerate}
As the clusters are non-overlapping it is immediate that $ T $ is indeed a tree.
To analyze the average stretch of~$ T $, we distinguish between inter-cluster edges (with endpoints in different clusters) and intra-cluster edges (with endpoints in the same cluster).
Each intra-cluster edge has stretch at most $ O (\tfrac{\log{n}}{\fract}) $ as the spanning tree of the cluster containing both endpoints of such an edge is a subtree of $ T $.
Each inter-cluster edge has polylogarithmic average stretch in $ T' $ with respect to $ G' $.
By expanding the clusters, the length of each path in $ T' $ increases by a factor of at most $ O (\tfrac{\log{n}}{\fract}) $.
Thus, inter-cluster edges have an average stretch of $ O (\tfrac{\log{n}}{\fract} \polylog{n}) $ in $ T $.
As there are at most $ m $ intra-cluster edges and at most $ \fract m $ inter-cluster edges, the total stretch over all edges is at most $ O (m \cdot \tfrac{\log{n}}{\fract} + \fract m \cdot \tfrac{\log{n}}{\fract} \polylog{n}) = \tilde O (m \cdot \tfrac{1}{\fract}) $, which gives an average stretch of $ \tilde O (\tfrac{1}{\fract}) $.

To bound the update time, first observe that the number of inter-cluster edges is at most $ \fract m $.
Thus, $ G' $ has at most $ \fract m $ edges and therefore the static algorithm for computing $ T' $ takes time $ \tilde O (\fract m) $ per update.
Together with the update time of the dynamic LDD, we obtain an update time of $ \tilde O (\tfrac{1}{\fract^2} + \fract m) $.
By setting $ \fract = m^{1/3} $, we would already obtain an algorithm for maintaining a tree of average stretch $ \tilde O (m^{1/3}) $ with update time $ \tilde O (m^{2/3}) $.

We can improve the stretch and still keep the update time sublinear by a hierarchical approach in which the scheme of clustering and contracting is repeated $ k $ times.
Observe that the $i$-th contracted graph will contain at most $ \fract^i m $ many edges and, in the final tree $ T $, the stretch of each edge disappearing with the $(i + 1)$-th contraction is $ O (\tfrac{\log{n}}{\fract})^{i + 1} $, which can be obtained by expanding the contracted low-diameter clusters. 
After $ k $~contractions, there are at most $ \fract^k m $ edges remaining and they have polylogarithmic average stretch in~$ T' $ with respect to~$ G' $, which, again by expanding clusters, implies an average stretch of at most $ O (\tfrac{\log{n}}{\fract})^k \cdot \polylog{n} $ in~$ T $ with respect to~$ G $.
This leads to a total stretch of $ O (\sum_{0 \leq i \leq k - 1} \fract^i m \cdot O(\tfrac{\log{n}}{\fract})^{i + 1} + \fract^k m \cdot O (\tfrac{\log{n}}{\fract})^k \polylog{n}) = \tilde O (m \cdot \tfrac{O (\log n)^k}{\fract}) $, which gives an average stretch of $ \tilde O (\tfrac{O (\log n)^k}{\fract}) $.
To bound the update time, observe that updates propagate within the hierarchy as each change to inter-cluster edges of one layer will appear as an update in the next layer.
Each operation in the dynamic LDD algorithm will perform at most one change to the clustering, i.e., the number of changes propagated to the next layer of the hierarchy is at most $ \tilde O(\tfrac{1}{\fract^2}) $ per update to the current layer.
This will result in an update time of $ \tilde O ((\tfrac{\polylog{n}}{\fract})^{2(i - 1)} \cdot \frac{1}{\fract^2}) $ in the $i$-th contracted graph per update to the input graph.
The update time for maintaining the tree~$ T $ will then be $ \tilde O (\tfrac{1}{\fract^{2 k}} + \fract^k m) $, which is $ m^{2/3} $ at best, i.e., no better than the simpler approach above.
A tighter analysis can improve this update time significantly:
The second part of Theorem~\ref{thm:fully dynamic LDD} bounds the amortized number of edges to become inter-cluster edges by $ \tilde O (\tfrac{1}{\fract}) $.
This results in an update time of $ \tilde O ((\tfrac{\polylog{n}}{\fract})^{k + 1} + \fract^k m) $.
By setting $ k = \sqrt{\log{n}} $ and $ \fract = \tfrac{1}{m^{1 / (2k + 1)}} $ we can roughly balance these two terms in the update time and thus arrive at an update time of $ m^{1/2 + o(1)} $ while the average stretch is $ n^{o(1)} $.
The crux of our approach is thus an ``early stopping'' of the Alon et al.\ LDD hierarchy such that it does not ``exhaust'' the graph.
We crucially exploit that, for an unweighted input graph, the size of the contracted graph decreases geometrically, which allows us to partially compensate for the blow-up of propagated updates in the hierarchy.

We can use the following sparsification approach to further reduce the update time to $ n^{1/2 + o(1)} $:
The main idea is to maintain a cut sparsifier with $ \tilde O (n) $ edges and then run the algorithm on the cut sparsifier instead of the input graph to reduce the update time from $ m^{1/2 + o(1)} $ to $ n^{1/2 + o(1)} $.
The dynamic algorithm of Abraham et al.~\cite{AbrahamDKKP16} can maintain such a cut sparsifier with polylogarithmic update time.
Using a different cut sparsifier construction, Koutis, Levin, and Peng~\cite{KoutisLP16} showed in the static setting that a low-stretch tree of their cut sparsifier is also a low-stretch tree of the input graph (where the average stretch only increases multiplicatively by the approximation guarantee of the cut sparsifier).
However, we cannot use exactly the same approach because the cut sparsifier of Abraham et al.\ has edge weights, even though the input graph is unweighted.
We show that the main argument in~\cite{KoutisLP16} still goes through if we interpret the edge weights of the sparsifier as edge multiplicities in an unweighted graph.
We then show that the algorithm of Theorem~\ref{thm:fully dynamic low stretch tree} can also handle such graphs for updates that increment or decrement the multiplicity of some edge by $ 1 $.
A fine-grained analysis of the total multiplicity of edges of the sparsifier and its expected amount of change per update to the input graph then gives the desired result.

In Section~\ref{sec:low_stretch_tree}, where we present the details of our approach, we consider two slight generalizations:
First, we implicitly handle the case that the input graph could become disconnected by maintaining a low-stretch \emph{forest}.
Second, we give a parameterized analysis that also allows for a trade-off between stretch and update time.

\paragraph*{Low-Diameter Decomposition.}

To obtain a suitable algorithm for dynamically maintaining a low-diameter decomposition, we follow the widespread paradigm of first designing a decremental -- i.e., deletions-only -- algorithm and then extending it to a fully dynamic one.
We can show that, for any sequence of at most $ m $ edge deletions (where $ m $ is the initial number of edges in the graph), a $ (\fract, O(\tfrac{\log{n}}{\fract})) $-decomposition can be maintained with expected total update time $ \tilde O (m / \fract) $.
Here, we build upon the work of Miller et al.~\cite{MillerPX13} who showed that \emph{exponential random-shift clustering} produces clusters of radius $ O (\log{n} / \fract) $ such that each edge has a probability of at most $ \fract $ to go between clusters.
This clustering is obtained by first having each node sample a random \emph{shift value} from the exponential distribution and then determining the cluster center of each node as the node to which it minimizes the difference between distance and (other node's) shift value.

In the parallel algorithm of~\cite{MillerPX13}, the clustering is obtained by essentially computing one single-source shortest path tree of maximum depth $ O (\log{n} / \fract) $.
To make this computation efficient\footnote{For their parallel algorithm, efficiency in particular means low depth of the computation tree.}, the shift values are rounded integers and the fractional residuals are only considered for tie-breaking.
We observe that one can maintain this bounded-depth shortest path tree with a simple modification of the well-known Even-Shiloach algorithm that spends time $ O ( \degree (v)) $ every time a node $ v $ increases its level (distance from the source) in the tree.
By rounding to integer edge weights, similar to~\cite{MillerPX13}, we can make sure that the number of level increases to consider is at most $ O (\log{n} / \fract) $ for each node.
Note however that this standard argument charging each node only when it increases its level is not enough for our purpose: the assignment of nodes to clusters follows the fractional values for tie-breaking, which might result in some node $ v $ changing its cluster -- and in this way also spend time $ O (\degree (v)) $ -- without increasing its level~(note that here the difficulty is not on maintaining the cluster that $v$ belongs to, but rather on bounding the number of cluster changes for $v$).
As has been observed in~\cite{MillerPX13}, the fractional values of the shift values effectively induce a random permutation on the nodes.
Using a similar argument as in the analysis of the dynamic spanner algorithm of Baswana et al.~\cite{BaswanaKS12}, we can thus show that in expectation each node changes its cluster at most $ O (\log{n}) $ times while staying at a particular level.
This results in a total update time of $ \tilde O (m / \fract) $.
Trivially, this also bounds the total number of times that edges become inter-cluster edges during the whole decremental algorithm by $ \tilde O (m / \fract) $.
Using a more sophisticated analysis we can obtain the stronger bound of $ \tilde O (m) $ on the latter quantity:
Intuitively, each endpoint of an edge changes its cluster at most $ \tilde O (\tfrac{1}{\fract}) $ times and after each cluster change the edge is an inter-cluster edge with probability at most $ \beta $, yielding a total of $ \tilde O ( m \cdot \tfrac{1}{\fract} \cdot \beta ) $ times that edges become inter-cluster edges.
The rigorous argument is however more complicated because the event of being an inter-cluster edge might not be independent from the event of the endpoint changing its cluster.

To obtain a fully dynamic algorithm, we observe that any LDD can tolerate a certain number of insertions to the graph.
A $ (\fract, O(\tfrac{\log{n}}{\fract})) $-decomposition allows at most $ \fract m $ inter-cluster edges and thus, if we insert $ O (\fract m) $ edges to the graph without changing the decomposition, we still have an $ (O(\fract), O(\tfrac{\log{n}}{\fract})) $-decomposition.
We can exploit this observation by simply running a decremental algorithm, that is restarted from scratch after each phase of $ \Theta (\fract m) $ updates to the graph.
We then deal with edge deletions by delegating them to the decremental algorithm and we deal with edge insertions in a lazy way by doing nothing.
This results in a total time of $ \tilde O (m / \fract) $ that is amortized over $ \Theta (\fract m) $ updates to the graph, i.e., amortized update time $ \tilde O (1/\beta^2) $.
Similarly, the amortized number of edges to become inter-cluster edges after an update is $ \tilde O (1/\beta) $.

In our detailed description and analysis in Section~\ref{sec:LDD}, we first review the construction of Miller et al., and then present our decremental and fully dynamic algorithms.

%% file: low_stretch_tree.tex
\section{Dynamic Low-Stretch Forest}\label{sec:low_stretch_tree}

Our dynamic algorithms for maintaining a low-stretch forest will use a hierarchy of low-diameter decompositions.
We first analyze very generally the update time for maintaining such a decomposition and explain how to obtain a spanning forest from this hierarchy in a natural way, similar to the construction of Alon et al.~\cite{AlonKPW95}.
We then analyze two different approaches for maintaining the tree, which will give us two complementary points in the design space of dynamic low-stretch tree algorithms.
Finally, we explain how to exploit input graph sparsification to improve the update time of our first algorithm.

\subsection{Generic Dynamic LDD Hierarchy}

Consider some integer parameter $ k \geq 1 $ and parameters $ \fract_0, \ldots, \fract_{k-1} \in (0, 1) $.
For each $ 0 \leq i \leq k-1 $, let $ \mathcal{D}_i $ be the fully dynamic algorithm for maintaining a $ (\fract_i, O (\tfrac{\log{n}}{\fract_i})) $-decomposition as given by Theorem~\ref{thm:fully dynamic LDD}.
Our \emph{LDD-hierarchy} consists of $ k+1 $ multigraphs $ G_0 = (V, E_0), \ldots, G_k = (V, E_k) $ where $ G_0 $ is the input graph $ G $ and, for each $ 0 \leq i \leq k-1 $, the graph $ G_{i+1} $ is obtained from contracting $ G_i $ according to a $ (\beta_i, O (\tfrac{\log{n}}{\fract_i}))$-decomposition of $ G_i $ as follows:
For every node $ v \in V $, let $ c_i (v) $ denote the center of the cluster to which $ v $ is assigned in the $ (\beta_i, O (\tfrac{\log{n}}{\fract_i}))$-decomposition of $ G_i $.
Now define $ E_{i+1} $ as the multiset of edges containing for every edge $ (u, v) \in E_i $ such that $ c_i (u) \neq c_i (v) $ one edge $ (c_i (u), c_i (v)) $, i.e., $ E_{i+1} = \{ (c_i (u), c_i (v)) : (u, v) \in E_i \text{ and } c_i (u) \neq c_i (v) \} $, where the multiplicity of each edge is equal to the number of edges between the corresponding clusters in~$ G_i $.
Remember that all graphs~$ G_i $ have the same set of nodes, but nodes that do not serve as cluster centers in $ G_{i-1} $ will be isolated in $ G_i $.
It might seem counter-intuitive at first that these isolated nodes are not removed from the graph, but observe that in our dynamic algorithm nodes might start or stop being cluster centers over time.
By keeping all nodes in all subgraphs, we avoid having to explicitly deal with insertions or deletions of nodes.\footnote{Note that it is easy to explicitly maintain the sets of isolated and non-isolated nodes by observing the degrees.}

Note that the $ (\fract_i, O (\tfrac{\log{n}}{\fract_i})) $-decomposition of $ G_i $ guarantees that $ | E_{i+1} | \leq \beta_i \cdot | E_i | $ in expectation, which implies the following bound.

\begin{observation}\label{obs:number of nodes and edges}
For every $ 0 \leq i \leq k $, $ |E_i| \leq m \cdot \prod_{0 \leq j \leq i - 1} \beta_j $ in expectation.\footnote{Note that for $ i = 0 $ the product $ \prod_{0 \leq j \leq i-1} \fract_j $ is empty and thus equal to $ 1 $.}
\end{observation}

We now analyze the update time for maintaining this LDD-hierarchy under insertions and deletions to the input graph~$ G $.
Note that for each level~$ i \leq k - 1 $ of the hierarchy, changes made to the graph $ G_i $ might result in the dynamic algorithm~$ \mathcal{D}_i $ making changes to the $ (\fract_i, O (\tfrac{\log{n}}{\fract_i})) $-decomposition of $ G_i $.
In particular, edges of $ G_i $ could start or stop being inter-cluster edges in the decomposition, which in turn leads to edges being added to or removed from~$ G_{i + 1} $.
Thus, a single update to the input graph~$ G $ might result in a blow-up of induced updates to be processed by the algorithms $ \mathcal{D}_1, \ldots, \mathcal{D}_{k -1} $.
To limit this blow-up, we use an additional property of our LDD-decomposition given in Theorem~\ref{thm:fully dynamic LDD}, namely the non-trivial bound on the number of edges to become inter-cluster edges after each update.

\begin{lemma}\label{lem:update time hierarchy}
The LDD-hierarchy can be maintained with an expected amortized update time of
\begin{equation*}
\tilde O \left( \sum_{0 \leq j \leq k-1} \frac{ O (\log n)^{2 (k-1)} }{\fract_j \prod_{0 \leq j' \leq j} \fract_{j'}} \right) \, .
\end{equation*}
\end{lemma}

\begin{proof}
For every $ 0 \leq i \leq k - 1 $ and every $ q \geq 1 $ define the following random variables:
\begin{itemize}
\item $ X_i (q) $: The total time spent by algorithm $ \mathcal{D}_i $ for processing any sequence of $ q $~updates to~$ G_i $.
\item $ Y_i (q) $: The total number of changes performed to $ G_{i+1} $ by $ \mathcal{D}_i $ while processing any sequence of $ q $~updates to~$ G_i $.
\item $ Z_i (q) $: The total time spent by algorithms $ \mathcal{D}_i, \ldots, \mathcal{D}_{k-1} $ for processing any sequence of $ q $~updates to~$ G_i $.
\end{itemize}
Note that the expected values of $ X_i (q) $ and $ Y_i (q) $ are bounded by Theorem~\ref{thm:fully dynamic LDD}~(the latter holds since only changes involving inter-cluster edges are propagated as updates to the next level).
We will show by induction on $ i $ that $ E [Z_i (q)] = \tilde O (q \cdot \sum_{i \leq j \leq k-1} \tfrac{ O (\log n)^{2(k-i-1)} }{\fract_j \prod_{i \leq j' \leq j} \fract_{j'}}) $, which with $ i = 0 $ implies the claim we want to prove.

Before showing the proof, observe that our LDD-hierarchy uses multiple instances of the dynamic low-diameter decomposition. We can order these instances in a hierarchical manner such that changes in the instance $i$ only affect instances $i+1$ and above (this is possible because all changes propagate one way through the hierarchy). Since the random bits among levels are independent, we can think of the random bits in the previous level being fixed in advance, and hence the updates to the instance $i$ are fixed as well. The latter implies that each instance $i$ in the LDD-hierarchy is running in the oblivious adversary setting, as required by Theorem~\ref{thm:fully dynamic LDD}.

We next prove the claimed bound on $E [Z_i (q)]$. In the base case $ i = k - 1 $, we know by Theorem~\ref{thm:fully dynamic LDD} that algorithm $ \mathcal{D}_{k-1} $ maintaining the $ (\fract_{k - 1}, O (\tfrac{\log{n}}{\fract_{k - 1}})) $-decomposition of $ G_{k-1} $ spends expected amortized time $ \tilde O (\frac{1}{\fract_{k-1}^2}) $ per update to $ G_{k-1} $, i.e., $ E [Z_{k - 1} (q)] = E [X_{k - 1} (q)] = \tilde O (q \cdot \frac{1}{\fract_{k-1}^2}) $ for any $ q \geq 1 $.
For the inductive step, consider some $ 0 \leq i < k - 1 $ and any $ q \geq 1 $.
Any sequence of $ q $ updates to $ G_i $ induces at most $ Y_i (q) $ updates to $ G_{i + 1} $.
Each of those updates has to be processed by the algorithms $ \mathcal{D}_{i+1}, \ldots, \mathcal{D}_{k-1} $.
We thus have $ Z_i (q) = X_i (q) + Z_{i+1} (Y_i (q)) $.

To bound $ E [Z_i (q)] $, recall first the expectations of the involved random variables.
As by Theorem~\ref{thm:fully dynamic LDD} the algorithm~$ \mathcal{D}_i $ maintaining the $ (\fract_i, O (\tfrac{\log{n}}{\fract_i})) $-decomposition of $ G_i $ has expected amortized update time $ \tilde O (\tfrac{1}{\fract_i^2}) $, it spends an expected total time of $ E [X_i (q)] = \tilde O (q \cdot \tfrac{1}{\fract_i^2}) $ for any sequence of $ q $ updates to $ G_i $.
Furthermore, over the whole sequence of $ q $~updates, the expected number of edges to ever become inter-cluster edges in the $ (\fract_i, O (\tfrac{\log{n}}{\fract_i})) $-decomposition of $ G_i $ is $ O (q \cdot \tfrac{ \log^2 n }{\fract_i}) $.
This induces at most $ O (q \cdot \tfrac{ \log^2 n }{\fract_i}) $ updates (insertions or deletions) to the graph $ G_{i+1} $, i.e., $ E [Y_i (q)] = O (q \cdot \tfrac{ \log^2 n }{\fract_i}) $.
By the induction hypothesis, the expected amortized update time spent by $ \mathcal{D}_{i+1}, \ldots, \mathcal{D}_{k-1} $ for any sequence of $ q' $ updates to $ G_{i+1} $ is $ E [Z_{i+1} (q')] = \tilde O (q' \cdot \sum_{i + 1 \leq j \leq k-1} \tfrac{ O (\log n)^{2(k-i-2)} }{\fract_j \prod_{i + 1 \leq j' \leq j} \fract_{j'}}) $.

Now by linearity of expectation we get
\begin{equation*}
E [Z_i (q)] = E \left[ X_i (q) + Z_{i+1} (Y_i (q)) \right] = E \left[ X_i (q) \right] + E \left[ Z_{i+1} (Y_i (q)) \right]
\end{equation*}
and by the law of total expectation we can bound $ E \left[ Z_{i+1} (Y_i (q)) \right] $ as follows:
\begin{align*}
E \left[ Z_{i+1} (Y_i (q)) \right] &= \sum_y E \left[ Z_{i+1} (Y_i (q)) \mid Y_i (q) = y \right] \cdot \Pr [Y_i (q) = y] \\
 &= \sum_y E \left[ Z_{i+1} (y) \right] \cdot \Pr [Y_i (q) = y] \\
 &= \sum_y \tilde O \left( y \cdot \sum_{i + 1 \leq j \leq k-1} \frac{ O (\log n)^{2(k-i-2)} }{\fract_j \prod_{i + 1 \leq j' \leq j} \fract_{j'}} \right) \cdot \Pr [Y_i (q) = y] \\
 &= \tilde O \left( \sum_{i + 1 \leq j \leq k-1} \frac{ O (\log n)^{2(k-i-2)} }{\fract_j \prod_{i + 1 \leq j' \leq j} \fract_{j'}} \right) \cdot \sum_y y \cdot \Pr [Y_i (q) = y] \\
 &= \tilde O \left( \sum_{i + 1 \leq j \leq k-1} \frac{ O (\log n)^{2(k-i-2)} }{\fract_j \prod_{i + 1 \leq j' \leq j} \fract_{j'}} \right) \cdot E [Y_i (q)] \\
 &= \tilde O \left( \sum_{i + 1 \leq j \leq k-1} \frac{ O (\log n)^{2(k-i-2)} }{\fract_j \prod_{i + 1 \leq j' \leq j} \fract_{j'}} \right) \cdot O \left( q \cdot \frac{ \log^2 n }{\fract_i} \right) \\
 &= \tilde O \left( q \cdot \sum_{i+1 \leq j \leq k-1} \frac{ O (\log n)^{2(k-i-1)} }{\fract_j \prod_{i \leq j' \leq j} \fract_{j'}} \right)
\end{align*}
We thus get
\begin{equation*}
E [Z_i (q)] = \tilde O (q \cdot \tfrac{1}{\fract_i^2}) + \tilde O \left( q \cdot \sum_{i+1 \leq j \leq k-1} \frac{ O (\log n)^{2(k-i-1)} }{\fract_j \prod_{i \leq j' \leq j} \fract_{j'}} \right) = \tilde O \left( q \cdot \sum_{i \leq j \leq k-1} \frac{ O (\log n)^{2(k-i-1)} }{\fract_j \prod_{i \leq j' \leq j} \fract_{j'}} \right)
\end{equation*}
as desired.
\end{proof}

Given any spanning forest $ T' $ of $ G_k $, there is a natural way of defining a spanning forest~$ T $ of~$ G $ from the LDD-hierarchy.
To this end, we first formally define the contraction of a node $ v $ of $ G $ to a cluster center $ v' $ of $ G_i $ (for $ 0 \leq i \leq k) $ as follows:
Every node $ v $ of $ G $ is contracted to itself in $ G_0 $, and, for every $ 1 \leq i \leq k $, a node $ v $ of $ G $ is contracted to $ v' $ in $ G_i $ if $ v $ is contracted to $ u' $ in $ G_{i-1} $ and $ c_{i-1} (u') = v' $.
Similarly, for every $ 0 \leq i \leq k $, an edge $ e = (u, v) $ of $ G $ is contracted to an edge $ e' = (u', v') $ of $ G_i $ if $ u $ is contracted to $ u' $ and $ v $ is contracted to $ v' $.
Now define $ T $ inductively as follows:
We let $ T_0 $ be the forest consisting of the spanning trees of diamteter $ O (\tfrac{\log{n}}{\fract_0}) $ of the clusters in the $ (\fract_0, O (\tfrac{\log{n}}{\fract_0}))$-decomposition of $ G_0 $.
For every $ 1 \leq i \leq k $, we obtain $ T_i $ from $ T_{i-1} $ and a $ (\fract_i, O(\tfrac{\log{n}}{\fract_i})) $-decomposition of $ G_i $ as follows: for every edge $ e' $ in a shortest path tree in one of the clusters, we include in $ T_i $ \emph{exactly one} edge $e$ of $G$ among the edges that are contracted to $ e' $ in $ G_i $.
Finally, $ T $ is obtained from $ T_k $ as follows: for every edge~$ e' $ in the spanning forest~$ T' $ of~$ G_k $, we include in~$ T $ the edge~$ e $ of~$ G $ contracted to~$ e' $ in~$ G_k $.
As the clusters in each decomposition are non-overlapping, we are guaranteed that $ T $ is indeed a forest.
Note that, apart from the time needed to maintain $ T' $, we can maintain $ T $ in the same asymptotic update time as the LDD-hierarchy (up to logarithmic factors).

We now partially analyze the stretch of $ T $ with respect to $ G $.

\begin{lemma}\label{lem:path length for nodes in same cluster}
For every $ 1 \leq i \leq k $, and for every pair of nodes $ u $ and~$ v $ that are contracted to the same cluster center in $ G_i $, there is a path from $ u $ to $ v $ in $ T $ of length at most $ \tfrac{O(\log{n})^i}{\prod_{0 \leq j \leq i-1} \fract_j} $.
\end{lemma}

\begin{proof}
The proof is by induction on~$ i $.
The induction base $ i = 1 $ is straightforward:
For $ u $ and $ v $ to be contracted to the same cluster center in $ G_1 $, they must be contained in the same cluster~$ C $ of the $ (\fract_0, O (\tfrac{\log{n}}{\fract_0})) $-decomposition of $ G_0 $ maintained by $ \mathcal{D}_0 $.
Remember that $ C $ has strong diameter at most $ O (\tfrac{\log{n}}{\fract_0}) $.
Thus, in the shortest path tree of $ C $ there is a path of length at most $ O (\tfrac{\log{n}}{\fract_0}) $ from $ u $ to $ v $ using edges of $ G_0 = G $.
By the definition of $ T $, this path is also present in $ T $.

For the inductive step, let $ 2 \leq i \leq k $ and let $ u' $ and $ v' $ denote the cluster centers to which $ u $ and $ v $ are contracted in $ G_{i - 1} $, respectively.
For $ u $ and $ v $ to be contracted to the same cluster center in $ G_i $, $ u' $ and $ v ' $ must be contained in the same cluster~$ C $ of the $ (\fract_{i - 1}, O (\tfrac{\log{n}}{\fract_{i - 1}})) $-decomposition of $ G_{i - 1} $ maintained by $ \mathcal{D}_{i - 1} $.
As $ C $ has strong diameter at most $ O (\tfrac{\log{n}}{\fract_{i - 1}}) $, there is a path $ \pi $ from $ u' $ to $ v' $ of length at most $ O (\tfrac{\log{n}}{\fract_{i - 1}}) $ in the shortest path tree of $ C $.
Let $ x_1, \ldots, x_t $ denote the nodes on $ \pi $, where $ x_1 = u' $ and $ x_t = v' $.
By the definition of our tree $T$ with respect to $ G $, there must exist edges $ (a_1, b_1), \ldots, (a_t, b_t) $ of~$ G $ such that
\begin{itemize}
\item $ (a_\ell, b_\ell) $ is contained in $ T $ for all $ 1 \leq \ell \leq t $,
\item $ u $ and $ a_1 $ are contracted to the same cluster center in $ G_{i-1} $,
\item $ b_t $ and $ v $ are contracted to the same cluster center in $ G_{i-1} $, and
\item $ b_\ell $ and $ a_{\ell+1} $ are contracted to the same cluster center in $ G_{i-1} $ for all $ 1 \leq \ell \leq t - 1 $.
\end{itemize}
By the induction hypothesis we know that for every $ 1 \leq \ell \leq t-1 $ there is a path of length at most $ \tfrac{O (\log{n})^{i - 1}}{\prod_{0 \leq j \leq i-2} \fract_j} $ from $ b_\ell $ to $ a_{\ell+1} $ in $ T $.
Paths of the same maximum length also exist from $ u $ to $ a_1 $ and from $ b_t $ to $ v $.
It follows that there is a path from $ u $ to $ v $ in $ T $ of length at most
\begin{align*}
(t + 1) \cdot \frac{O (\log{n})^{i - 1}}{\prod_{0 \leq j \leq i-2} \fract_j} + t \leq 3 t \cdot \frac{O (\log{n})^{i - 1}}{\prod_{0 \leq j \leq i-2} \fract_j}
 = O \left( \frac{\log{n}}{\fract_i} \right) \cdot \frac{O (\log{n})^{i - 1}}{\prod_{0 \leq j \leq i-2} \fract_j}
 = \frac{O (\log{n})^i}{\prod_{0 \leq j \leq i-1} \fract_j}
\end{align*}
as desired.
\end{proof}

To analyze the stretch of $ T $, we will use the following terminology:
we let the \emph{level} of an edge $ e $ of $ G $ be the largest $ i $ such that edge $ e $ is contracted to some edge~$ e' $ in $ G_i $.
Remember that $ E_i $ is a multiset of edges containing as many edges $ (u', v') $ as there are edges $ (u, v) \in E $ with $ u $ and $ v $ being contracted to different cluster centers $ u' $ and $ v' $ in $ G_i $, respectively.
Thus, the expected number of edges at level $ i $ is at most $ | E_i | $.
Note that for an edge $ e = (u, v) $ to be at level~$ i $, $ u $ and $ v $ must be contracted to the same cluster center in $ G_{i + 1} $.
Therefore, by \Cref{lem:path length for nodes in same cluster}, the stretch of edges at level~$ i $ in $ T $ with respect to $ G $ is at most $ \frac{O(\log{n})^{i+1}}{\prod_{0 \leq j \leq i} \fract_j} $.
The expected contribution to the total stretch of~$ T $ by edges at level $ i \leq k - 1 $ is thus at most
\begin{equation}\label{eq:stretch of level i edges}
| E_i | \cdot \frac{O(\log{n})^{i+1}}{\prod_{0 \leq j \leq i} \fract_j} \leq \frac{m}{\beta_i} \cdot O (\log{n})^{i+1} \, .
\end{equation}

\subsection{Dynamic Low-Stretch Tree Algorithms}

To now obtain a fully dynamic algorithm for maintaining a low-stretch forest, it remains to plug in a concrete algorithm for maintaining $ T' $ together with suitable choices of the parameters.
We analyze two choices for dynamically maintaining $ T' $.
The first is the ``lazy'' approach of recomputing a low-stretch forest from scratch after each update to the input graph.
The second is a fully dynamic spanning forest algorithm with only trivial stretch guarantees.

\begin{theorem}[Restatement of Theorem~\ref{thm:fully dynamic low stretch tree}]
Given any unweighted, undirected graph undergoing edge insertions and deletions, there is a fully dynamic algorithm for maintaining a spanning forest of expected average stretch $ n^{o(1)} $ that has expected amortized update time $ m^{1/2 + o(1)} $.
These guarantees hold against an oblivious adversary.
\end{theorem}

\begin{proof}

We set $ k = \lceil \sqrt{\log{n}} \rceil $ and $ \fract_i = \fract = \tfrac{1}{m^{1 / (2k + 1)}} $ for all $ 0 \leq i \leq k-1 $ and maintain an LDD-hierarchy with these parameters.
Additionally, we maintain the graph $ G' $ induced by all non-isolated nodes of $ G_k $, which can easily be done by maintaining the degrees of nodes in $ G_k $.
After each update to $ G $, we compute a low-average stretch forest of $ T' $ of $ G' $.
Note that this recomputation is performed \emph{after} having updated all graphs in the hierarchy; we use the state-of-the-art static algorithm for computing a spanning forest of the multigraph $ G' $ with total stretch $ \tilde O (| E_k |) $ in time $ \tilde O (| E_k |) $.

By Equation~\eqref{eq:stretch of level i edges}, the contribution to the total stretch of $ T $ by edges at level $ i \leq k - 1 $ is at most $ m \cdot \tfrac{O (\log{n})^{i+1}}{\beta_i} $.
To bound the contribution of edges at level~$ k $, consider some edge $ e = (u, v) $ at level~$ k $ and let $ u' $ and $ v' $ denote the cluster centers to which $ u' $ and $ v' $ are contracted in $ G_k $, respectively.
Let $ \pi $ denote the path from $ u' $ to $ v' $ in $ T' $.
Using similar arguments as in the proof of Lemma~\ref{lem:path length for nodes in same cluster}, the contracted nodes and edges of $ \pi $ can be expanded to a path from $ u $ to~$ v $ in~$ T $ of length at most $ \tfrac{O (\log{n})^k}{\prod_{0 \leq i \leq k - 1} \fract_i} \cdot | \pi | $.
Thus, the contribution of edges at level~$ k $ is at most $ \tilde O (| E_k |) \cdot \tfrac{O (\log{n})^k}{\prod_{0 \leq i \leq k - 1} \fract_i} = \tilde O (m \cdot O (\log{n})^k) $ and the total stretch of~$ T $ with respect to~$ G $ is
\begin{align*}
\sum_{0 \leq i \leq k - 1} m \cdot \frac{O (\log{n})^{i+1}}{\beta} + \tilde O (m \cdot O (\log{n})^k)
 &= \tilde O \left( m \cdot \left( \frac{1}{\fract} \cdot \sum_{0 \leq i \leq k - 1} O (\log{n})^{i+1} + O (\log{n})^k \right) \right) \\
 &= \tilde O \left( m \cdot \frac{O (\log{n})^k}{\fract} \right) \\
 &= \tilde O \left( m m^{1 / (2k + 1)} \cdot O (\log{n})^k \right) \\
 &= m^{1 + o(1)} \, ,
\end{align*}
which gives an average stretch of $ m^{o(1)} = n^{o(1)} $.

By Observation~\ref{obs:number of nodes and edges}, $ G_k $ has at most $ m \beta^k $ edges in expectation and thus $ G' $ has at most $ m \beta^k $ nodes and edges in expectation.
Using the bound of Lemma~\ref{lem:update time hierarchy} for the update time of the LDD-hierarchy and the bound of $ \tilde O (m \beta^k ) $ for recomputing the low-stretch tree $ T' $ on $ G' $ from scratch, the expected amortized update time for maintaining $ T $ is
\begin{align*}
\tilde O \left( \sum_{0 \leq j \leq k-1} \frac{O (\log n)^{2 k}}{\fract_j \prod_{0 \leq j' \leq j} \fract_{j'}} + | E_k | \right) &= \tilde O \left( \sum_{0 \leq j \leq k-1} \frac{O (\log n)^{2 k}}{\fract^{j + 2}} + m \beta^k \right) \\
 &= \tilde O \left( \frac{O (\log n)^{2 k}}{\fract^{k + 1}} + m \beta^k \right) \\
 &= \tilde O (m^{(k + 1) / (2k + 1)} \cdot O (\log n)^{2 k}) \\
 &= \tilde O (m^{1/2 + 1 / (4k + 2)} \cdot O (\log n)^{2 k}) = m^{1/2 + o(1)} \, . \qedhere
\end{align*}
\end{proof}

\begin{theorem}\label{thm:LST trade-off}
Given any unweighted, undirected graph undergoing edge insertions and deletions, there is a fully dynamic algorithm for maintaining a spanning forest of expected average stretch $ O (t + n^{1/3 + o(1)}) $ that has expected amortized update time $ \tfrac{n^{1 + o(1)}}{t} $ for every $ 1 \leq t \leq n $.
These guarantees hold against an oblivious adversary.
\end{theorem}

\begin{proof}
We set $ k = \lceil \log{\log{n}} \rceil $, $ \fract_0 = \sqrt{t/n} $ and $ \fract_i = \sqrt{\fract_{i-1}} $ for all $ 1 \leq i \leq k-1 $ and maintain an LDD-hierarchy with these parameters.
The spanning forest $ T' $ is obtained by fully dynamically maintaining a spanning forest of $ G_k $ using any algorithm with polylogarithmic update time. 

By Equation~\eqref{eq:stretch of level i edges}, the contribution to the total stretch of $ T $ by edges at level $ i \leq k - 1 $ is at most $ \frac{m}{\beta_i} \cdot O (\log{n})^{i+1} $.
For every edge $ e = (u, v) $ at level~$ k $ with $ u $ contracted to $ u' $ and $ v $ contracted to $ v' $ in $ G_k $, there is a path from $ u' $ to $ v' $ in $ T' $ that by undoing the contractions can be expanded to a path from $ u $ to $ v $ in $ T $, which trivially has length at most $ n - 1 $.
Thus, the contribution by each edge at level $ k $ is at most $ n - 1 $.
As for every $ 0 \leq i \leq k $ there are at most $ | E_i | = m \cdot \prod_{0 \leq j \leq i - 1} \beta_j $ edges at level $ i $ in expectation, we can bound the expected total stretch of $ T $ with respect to $ G $ as follows:
\begin{align*}
\sum_{0 \leq i \leq k-1}  | E_i | \cdot \frac{O(\log{n})^{i+1}}{\prod_{0 \leq j \leq i} \fract_j} + | E_k | \cdot n &=
\sum_{0 \leq i \leq k-1} \frac{m \cdot O (\log{n})^{i+1}}{\fract_i} + m \cdot \prod_{0 \leq i \leq k-1} \fract_i \cdot n \\
 &= m \cdot \left( \sum_{0 \leq i \leq k-1} \frac{O (\log{n})^{i+1}}{\fract_i} + \prod_{0 \leq i \leq k-1} \fract_i \cdot n \right)
\end{align*}
This gives an average stretch of $ \sum_{0 \leq i \leq k-1} \tfrac{O (\log{n})^{i+1}}{\fract_i} + \prod_{0 \leq i \leq k-1} \fract_i \cdot n $.
We now simplify these two terms.
Exploiting that $ \fract_i \geq \fract_0 $ for all $ 1 \leq i \leq k-1 $, we get
\begin{equation*}
\sum_{0 \leq i \leq k-1} \frac{O (\log{n})^{i+1}}{\fract_i} \leq \sum_{0 \leq i \leq k-1} \frac{O (\log{n})^{i+1}}{\fract_0} = \frac{O (\log{n})^k}{\fract_0} = \frac{O (\log n)^k}{\sqrt{t/n}} = \sqrt{\frac{n^{1+o(1)}}{t}} \, .
\end{equation*}
Furthermore, the geometric progression of the $ \fract_i $'s gives
\begin{equation*}
\prod_{0 \leq i \leq k-1} \fract_i \cdot n = \prod_{0 \leq i \leq k-1} \fract_0^{1/2^i} \cdot n = \fract_0^{\sum_{0 \leq i \leq k-1} 1/2^i} \cdot n = \fract_0^{2 - 1/2^{k-1}} \cdot n = \frac{t^{1 - 1/2^k}}{n^{1 - 1/2^k}} \cdot n \leq t \cdot n^{1/2^k} = O (t) \, .
\end{equation*}
The average stretch of the forest maintained by our algorithm is thus at most $ O ( t + \sqrt{\tfrac{n^{1+o(1)}}{t}} ) $, which, after balancing the two terms, can be rewritten as $ O (t + n^{1/3 + o(1)}) $.

It remains to bound the update time of the algorithm.
By Lemma~\ref{lem:update time hierarchy}, the hierarchy can be maintained with an amortized update time of
\begin{align*}
\tilde O \left( \sum_{0 \leq j \leq k-1} \frac{O (\log n)^{2 k}}{\fract_j \cdot \prod_{0 \leq j' \leq j} \fract_{j'}} \right) = \tilde O \left( \sum_{0 \leq j \leq k-1} \frac{O (\log n)^{2 k}}{\fract_0^{1 / 2^j} \cdot \fract_0^{2 - 1 / 2^j}} \right) &= \tilde O \left( \sum_{0 \leq j \leq k-1} \frac{O (\log n)^{2 k}}{\fract_0^2} \right) \\ &= \frac{n \cdot O (\log n)^{2 k}}{t} = \frac{n^{1 + o(1)}}{t} \, .
\end{align*}
Since the amortized number of changes to $ G_k $ per update to $ G $ is trivially bounded by $ \tfrac{n^{1 + o(1)}}{t} $ as well and since $ T' $ can be maintained with polylogarithmic amortized time per update to $ G_k $, we can maintain $ T $ with amortized update time $ \tfrac{n^{1 + o(1)}}{t} $.
\end{proof}

Note that the algorithm of Theorem~\ref{thm:fully dynamic low stretch tree} is superior to the algorithm of Theorem~\ref{thm:LST trade-off} as long as $ t \leq \sqrt{n} $.
If $ t \geq \sqrt{n} $, then the algorithm of Theorem~\ref{thm:LST trade-off} provides stretch $ O (t) $ and update time $ \tfrac{n^{1 + o(1)}}{t} $.

\subsection{Input Graph Sparsification}\label{apx:sparsifier}

In the following, we explain how input graph sparsification can be performed to the algorithm of Theorem~\ref{thm:fully dynamic low stretch tree} by running the algorithm on a cut sparsifier, similar to the approach of Koutis et al.~\cite{KoutisLP16} in the \emph{static} setting.

\begin{corollary} [Restatement of Corollary~\ref{cor:fully dynamic low stretch tree}] \label{corApp:fully dynamic low stretch tree}
Given any unweighted, undirected graph undergoing edge insertions and deletions, there is a fully dynamic algorithm for maintaining a spanning forest of expected average stretch $ n^{o(1)} $ that has expected amortized update time $ n^{1/2 + o(1)} $.
These guarantees hold against an oblivious adversary.
\end{corollary}

To make the analysis rigorous, we introduce some additional notation for multigraphs.

\paragraph{Succinct Representation of Multigraphs.}
A multigraph $ G = (V, E) $ consists of a set of nodes~$ V $ and a multiset of edges~$ E $.
We denote by $ \bar{E} = \{ (u, v) \in \binom{V}{2} \mid (u, v) \in E \} $ the support of the multiset~$ E $.
This allows a multigraph $ G = (V, E) $ to be succinctly represented as its \emph{skeleton} $ \bar{G} = (V, \bar{E}, \mu_G) $ where $ \mu_G $ is a multiplicity function $ \mu_G : \bar{E} \rightarrow \mathbb{Z}^+ $ that assigns to each edge $ e $ its (positive integer) multiplicity $ \mu_G (e) $.
We denote by $ m \coloneqq | E | $ the number of multi-edges (considering multiplicities), and by $ \bar{m} \coloneqq | \bar{E} | $ the size of the support of $ E $ (disregarding multiplicities).
For simplicity, we assume that $ m $ is polynomial in~$ n $.
The total stretch of a spanning forest~$ T $ is defined with respect to~$ E $, i.e.,
\begin{equation}
\str_T (G) = \sum_{e = (u, v) \in E(G)} \dist_T (u, v) = \sum_{e = (u, v) \in \bar{E}(G)} \mu_G (e) \cdot \dist_T (u, v) \, . \label{eq:definition of total stretch}
\end{equation}

Our dynamic algorithm will exploit that, given the skeleton of a multigraph, a low-stretch forest can be computed without (significant) dependence on the multiplicities.
\begin{lemma}\label{lem:multiplicity oblivious running time}
Given the skeleton $ \bar{G} $ of a multigraph $ G $, a spanning forest of $ G $ of total stretch $ m^{1 + o(1)} $ can be computed in time $ \tilde O (\bar{m}) $.
\end{lemma}
Such a guarantee can be achieved with a \emph{static} version of our algorithm, i.e., by combining the scheme of Alon et al.~\cite{AlonKPW95} with the LDD of Miller et al.~\cite{MillerPX13}.
Although we are not aware of any statement of such a ``multiplicity-oblivious'' running time in the literature, it seems plausible that the state-of-the art algorithms (achieving a total stretch of $ \tilde O (m) $) also have this property.
Note however that a stretch of $ m^{1 + o(1)} $ is anyway good enough for our purpose.

\paragraph{Refined Analysis of Dynamic Low-Stretch Tree Algorithm.}

We now restate the guarantees of our fully dynamic low-stretch forest algorithm when the input is a multigraph undergoing insertions and deletions of multi-edges (i.e., each update increases or decreases the multiplicity of some edge by~$ 1 $).
Our fully dynamic LDD algorithm maintains a clustering such that every edge is an inter-cluster edge with probability $ \beta $.
This implies that at most a $ \beta $-fraction of the edges are inter-cluster edges in expectation -- regardless of whether we consider multiplicities.
More precisely, contracting the clusters to single nodes yields a multigraph $ G' = (V', E') $ with $ | E' | \leq \beta | E | $ and $ | \bar{E}' | \leq \beta | \bar{E} | $.
Now, in particular the LDD hierarchy in the proof of Theorem~\ref{thm:fully dynamic low stretch tree} results in a multigraph $ G' = (V', E') $ with $ | E' | \leq \beta^k m $ and $ | \bar{E}' | \leq \beta^k \bar{m} $ (after $k$ levels).
For such a graph, if its skeleton is given explicitly, one can compute a spanning forest of total stretch $ O (|E'|^{1 + o(1)}) $ in time $ \tilde O (|\bar{E}'|) $ by Lemma~\ref{lem:multiplicity oblivious running time}.
Note that our dynamic algorithm can explicitly maintain the skeleton of $ G' $ with neglegible overheads in the update time.
It follows that our algorithm maintains a spanning forest of total stretch $ O (m^{1 + o(1)}) $ and has an update time of $ \tilde O (\bar{m}^{1/2 + o(1)}) $.

\paragraph{Cut Sparsifiers.}
For the definition of cut sparsifiers, we consider cuts of the form $ (U, V \setminus U) $ induced by a subset of nodes $ U \subset V $.
The capacity of such a cut $ (U, V \setminus U) $ in a graph $ G $ is defined as the total multiplicity of edges crossing the cut, i.e.,
\begin{equation*}
\capacity_G (U, V \setminus U) = \sum_{\substack{e = (u, v) \in \bar{E} \\ u \in U, v \in V \setminus U}} \mu_G (e)
\end{equation*}
A $ (1 \pm \epsilon) $-\emph{cut sparsifier}~\cite{BenczurK15} (with $ 0 \leq \epsilon \leq 1/2 $) of a multigraph $ G = (V, E) $ is a ``subgraph'' $ H = (V, F) $ with $ \bar{F} \subseteq \bar{E} $ such that for every $ U \subset V $ we have
\begin{equation*}
(1 - \epsilon) \capacity_G (U, V \setminus U) \leq \capacity_H (U, V \setminus U) \leq (1 + \epsilon) \capacity_G (U, V \setminus U) \, ,
\end{equation*}
i.e., $ H $ approximately preserves all cuts of $ G $.
Now let $ H $ be a $ (1 \pm \epsilon) $-\emph{cut sparsifier} of a multigraph $ G = (V, E) $ and let $ T = (V, E(T))$ be a (simple) spanning forest of $ H $.
For every edge $ e $ of the forest~$ T $, the nodes are naturally partitioned into two connected subsets upon removal of~$ e $.
Let these two subsets be denoted by $ V_e $ and $ V \setminus V_e $.
Emek~\cite{Emek11} and Koutis et al.~\cite{KoutisLP16}, observed that by rearranging the sum in~\eqref{eq:definition of total stretch}, one obtains the following cut-based characterization of the stretch:
\begin{equation*}
\str_T (G) = \sum_{e \in E(T)} \capacity_G (V_e, V \setminus V_e) \, .
\end{equation*}
Observe that the cut $ (V_e, V \setminus V_e) $ is approximately preserved in $ H $ and thus $ \capacity_G (V_e, V \setminus V_e) \leq \tfrac{1}{1 - \epsilon} \capacity_H (V_e, V \setminus V_e) \leq (1 + 2 \epsilon) \capacity_H (V_e, V \setminus V_e) $.
The stretch of $ G $ with respect to~$ T $ can now be bounded by
\begin{align*}
\str_T (G) &= \sum_{e \in E(T)} \capacity_G (V_e, V \setminus V_e) \\
&\leq (1 + 2 \epsilon) \sum_{e \in E(T)} \capacity_H (V_e, V \setminus V_e) \\
&= (1 + 2 \epsilon) \str_T (H) \, .
\end{align*}
Thus, computing the low-stretch forest on the sparsifier $ H $ instead of the original graph $ G $ only increases the total stretch by a constant factor if the number of multi-edges in $ H $ is proportional to the number of edges in $ G $.

\paragraph{Dynamic Cut Sparsifiers.}
The fully dynamic algorithm of Abraham et al.~\cite{AbrahamDKKP16} maintains, with high probability, a $ (1 \pm \epsilon) $-cut sparsifier~$ H = (V, F) $ of a simple graph~$ G = (V, E) $ such that $ | \bar{F} | = \tilde O (n / \epsilon^2) $ with update time $ \poly (\log n, \epsilon) $.
For each node $ v $, the degree in $ H $ exceeds the degree in $ G $ by at most a factor of $ (1 \pm \epsilon) $ because the cut $ (\{ v \}, V \setminus \{ v \}) $ is approximately preserved in $ H $.
We can thus bound the number of multi-edges in $ H $ (i.e., the sum of all edge multiplicities) by $ | F | = O ((1 + \epsilon) |E|) $.
The algorithm maintains a hierarchy of the edges with $ O (\log n) $ layers, where edges at level~$ i $ have multiplicity $ 4^i $ and each edge is at level~$ i $ with probability at most $ 1/4^i $.
After an update to the input graph, the dynamic algorithm adds or removes at most $ \poly(\log n, \epsilon) $ edges in each level.
Thus, we can bound the amount of change to $ H $ per update to $ G $ as follows: for every update to $ G $, the expected sum of the changes to the edge multiplicities of~$ H $ is at most $ \poly(\log n, \epsilon) $.

\paragraph{Putting Everything Together~(Proof of Corollary~\ref{corApp:fully dynamic low stretch tree}).}
We now first use the fully dynamic algorithm of Abraham et al. to maintain a cut sparsifier $ H = (V, F) $ of the input graph $ G = (V, E) $ (with $ \epsilon = 1/2 $) and second run our fully dynamic low-stretch tree algorithm on top of $ H $.
Here, $ G $ is a simple graph with $ m = | E | $ edges and $ H $ is a multigraph with $ |F| = O (m) $ and $ |\bar{F}| = \tilde O (n) $.
The spanning forest $ T $ maintained in this way gives expected total stretch at most $ |F|^{1 + o(1)} $ with respect to~$ H $.
As argued above, this implies an expected total stretch of at most $ O ((1 + 2 \epsilon) |F|^{1 + o(1)}) = O (m^{1 + o(1)}) $ with respect to~$ G $, i.e., an average stretch of $ m^{o(1)} = n^{o(1)} $.
Each update to the input graph results in $ \polylog{n} $ changes to the sparsifier in expectation, which are then processed as ``induced'' updates by our dynamic low-stretch tree algorithm.
Thus, we overall arrive at an expected amortized update time of $ \tilde O (|\bar{F}|^{1/2 + o(1)}) = O (n^{1/2 + o(1)}) $.

%% file: LDD.tex
\section{Dynamic Low-Diameter Decomposition}\label{sec:LDD}

In this section we develop our dynamic algorithm for maintaining a low-diameter decomposition following three steps.
First, we review the static algorithm for constructing a low-diameter decomposition using the clustering due to Miller et al.~\cite{MillerPX13}.
Second, we design a decremental algorithm by extending the Even-Shiloach algorithm~\cite{EvenS81} in a suitable way.
Third, we lift our decremental algorithm to a fully dynamic one by using a ``lazy'' approach for handling insertions.

\input{static}

\subsection{Decremental Low-Diameter Decomposition}\label{sec:ES tree}\label{sec:decremental LDD}

We now show how to maintain a lower-diameter decomposition under deletion of edges. Recall that in the previous section we observed that computing a lower-diameter decomposition of a undirected, unweighted graph can be reduced to the single-source shortest path problem in some modified graph. In the same vein, we observe that maintaining a low-diameter decomposition under edge deletions amounts to maintaining a bounded-depth single-source shortest path tree of some modified graph under edge deletions. 

Even and Shiloach~\cite{EvenS81} devised a data-structure for maintaining a bounded-depth SSSP-tree under edge deletions, which we refer to as \emph{ES-tree}. The ES-tree initially worked only for undirected, unweighted graphs. However, later works~\cite{HenzingerK95,King99} observed that it can be extended even to directed, weighted graphs with positive integer edges weights. The mere usage of the ES-tree as a sub-routine will not suffice for our purposes, due to the constraints that our clustering imposes. In the following we show how to augment and modify an ES-tree that maintains a valid clustering, without degrading its running time guarantee. 

Let $G=(V,E)$ be an undirected, unweighted graph for which we want to maintain a decremental $(\beta, \log n/ \beta)$ decomposition, for any parameter $\beta \in (0,1)$. Further, let $G'=(V\cup \{s\},E',w')$ be the undirected graph with integral edge weights, as defined in Section~\ref{sec:staticLDD}.
Let $\pi$ be a random permutation on $V$. By discussion in Section~\ref{sec:staticLDD}, in order to maintain a low-diameter decomposition of $G$ it suffices to maintain a clustering of $G'$ with $\pi$ used for tie-breaking.

We describe an ES-tree that efficiently maintains a clustering of $G'$ for a given root node $s$ and a given distance parameter $\Delta$. Here we set $\Delta = O(\log n / \beta)$, as by Theorem~\ref{thm: Miller_LDD}, the maximum distance that we run our algorithm to is bounded by $O(\log n / \beta)$. Our data-structure handles arbitrary edge deletions, and maintains the following information. First, for each node $u \in V \cup \{s\}$, we maintain a label $\lev(u)$, referred to as the \emph{level} of $u$. This level of $u$ represents the shortest path between the root~$s$ and~$u$, i.e., $\lev(u) = \dist(s,u)$. Next, for each node $u \in V$, we maintain pointers $p(u)$ and $c(u)$, which represent the parent of $u$ in the tree and the node that $u$ is assigned to, respectively. Finally, we also maintain the set of \emph{potential} parents $P(u)$, for each $u \in V$, which is the set of all neighbors of $u$ that are in the same level with the parent of $u$, and share the same clustering with $u$, i.e., a neighbor~$v$ of~$u$ belongs to $P(u)$ if $v$ minimizes $(\ell(v) + w'(u,v), \pi(c(v)))$ lexicographically, and $c(v) = c(u)$. Edge deletions in $G'$ can possibly affect the above information for several nodes. Our algorithm adjusts these information on the nodes so as to make them valid for the modified graph.

\paragraph*{Algorithm Description and Implementation.}
We give an overview and describe the implementation of Algorithm~\ref{alg:ES_tree}. The data-structures $\ell(\cdot)$, $p(\cdot)$ and $c(\cdot)$ are initialized using Algorithm~\ref{alg:ModDijkstra} in Section~\ref{sec:staticLDD}. Note that for each $u \in V$, $P(u)$ can be computed by simply considering all neighbors of~$u$ in turn, and adding a neighbor $v$ to $P(u)$ if $v$ is a potential parent. The algorithm also maintains a heap $Q$ whose intended use is to store nodes whose levels or clustering might need to be updated. (see procedure \Initialize{}). 

In our decremental algorithm, each node tries to maintain its level $\ell(u)$, which corresponds to its current distance to the root $s$, together with its cluster pointer $c(u)$ in the current graph. Concretely, we maintain the following invariant for each node $u \in V$:
\begin{align}
\ell(u)&=\min\set{\ell(v)+w'(u, v)}{v \text{ is a neighbor of } u} \label{eq:ES tree level condition}
\end{align}
where ties among neighbors are broken according to~(\ref{eq: breakingTies}). This invariant allows to compute the cluster pointer $c(u)$ using Observation~\ref{obs: clusterUpdate}. Deleting an edge incident to $u$ might lead to a change in the values of $\ell(u)$ and $c(u)$. If this occurs, all neighbors of $u$ are notified by $u$ about this change, since their levels and cluster points might also change. It is well-known that the standard ES-tree can efficiently deal with changes involving the levels $\ell(\cdot)$. However, in our setting, it might be the case that an edge deletion forces a node $u$ to change its cluster while the level $\ell(u)$ still remains the same under this deletion. This is the point where our algorithm differs from the standard ES-tree, and we next show that (1) such changes can be handled efficiently, and (2) the number of cluster changes per node, within the same level, is small in expectation.

Let us consider the deletion of an edge $(u,v)$ (see procedure $\Delete()$); assume without loss of generality that $\ell(v) \leq \ell(u)$. Now note that an edge deletion might lead to a cluster change only if $v \in P(u)$. If this is the case, the algorithm first removes $v$ from the set $P(u)$. If $P(u)$ is still non-empty, the clustering remains unaffected. However, if $P(u)$ is empty, the clustering of $u$ will change, and the algorithm inserts $u$ into the heap $Q$ with key $\ell(u)$. Observe that a change in clustering of $u$ might potentially lead to cluster changes for children of $u$, given that $u$ was their only potential parent. In this way, we observe that deleting $(u,v)$ might force changes in the clustering for many descendants of $v$. The algorithm handles such changes using procedure $\UpdateLevels()$, which we describe below.

Procedure $\UpdateLevels()$ considers the nodes in $Q$ in the order of their current level. At each iteration, it takes the node $y$ with the smallest level $\ell(y)$ from $Q$. The node $y$ computes the set of potential parents $P(y)$, by examining each neighbor of $y$ in turn, and then adding to $P(y)$ all neighbors $z$ that minimize $(\ell(z) + w'(y,z), \pi(c(z)))$ lexicographically. Next, $y$ sets $p(y)$ as one of nodes in $P(y)$, and updates its level by setting $\ell(y) = \ell(p(y)) + w'(y,p(y))$. Having computed its parent pointer, $y$ updates the cluster pointer using Observation~\ref{obs: clusterUpdate}. Specifically, if the parent of $y$ is the source node $v$, then $y$ form a new cluster itself, i.e., $c(y) = y$. Otherwise, $y$ shares the same cluster with its parent and sets $c(y) = c(p(y))$. Finally, the algorithm determines whether the change in the clustering of $y$ affected its neighbors. Concretely, for each neighbor $x$ of $y$, it checks whether $y \in P(x)$. If this is not the case, then there is no change in the clustering of $x$. Otherwise, $y$ is removed from $P(x)$, and if $P(x)$ becomes empty after this removal, the algorithm inserts $x$ into the heap $Q$ with key $\ell(x)$, given that $Q$ does not already contain $x$. 

\paragraph*{Running Time Analysis.}

We first concern ourselves with the number of cluster change per node in our decremental algorithm. For any node $v \in V$, we say that the clustering \emph{changes} for~$v$ due to an edge deletion if this deletion either increases the level $\ell(v)$ or forces a change in the cluster pointer $c(v)$. It is well-known that the ES-tree can handle a level increase for any node $v$ in time $O(\deg(v))$. As we will see next, we can also handle a cluster change for a node in the same level in $O(\deg(v) \log n )$ time. However, we need to ensure that the number of such cluster changes for any node and any fixed level is small, for our algorithm to be efficient. Below we argue that one can have a fairly good bound on the expected number of such changes, and this is due to the special tie-breaking scheme we use when assigning nodes to clusters.

Fix any node $v \in V$, and consider $v$ during the sequence of edge deletions. Note that since only deletions are allowed, the level $\ell(v)$ is non-decreasing. This induces a natural partitioning of the sequence of edge deletions into subsequences such that the $\ell(v)$ remains unaffected during each subsequence. Specifically, for every node $v \in V$ and every $0 \leq i \leq \Delta$, let $S(i)$ the be subsequence of edge deletions during which $\ell(v) = i$, where $\Delta \leq O(\log n / \beta)$. The following bound on the expected number of cluster changes of $v$ during $S(i)$ follows an argument by Baswana et al.~\cite{BaswanaKS12}.

\begin{lemma} \label{lem: nrClusterChanges}
For every node $ v \in V$ and every $0 \leq i \leq \Delta$, during the entire subsequence $S(i)$, the cluster $c(v)$ of $v$ changes at most $O(\log n)$ times, in expectation.
\end{lemma}
\begin{proof}
Let $N_{i-1}(v)$ be the neighbors of $v$ at level $(i-1)$, grouped according to the the clusters they belong to. This grouping naturally induces a family $\mathcal{P}$ of all potential parents sets $P$ of $v$ at level $(i-1)$, just before the beginning of subsequence $S(i)$.
Let $C$ be the set of the corresponding clusters centers, i.e., for each $P \in \mathcal{P}$ add $c(P)$ to $C$, and note that $v$ can only join those centers during $S(i)$. Since we are considering only edge deletions, observe that when $v$ leaves a cluster centred at some node $c \in C$, it cannot join later the same cluster $c$ during $S(i)$.  

We next bound the number of cluster changes. For each $c \in C$, there must exist an edge in the subsequence $S(i)$ whose deletion increases $\dist(v,c)$, and thus $c$ is no longer a valid cluster center for $v$ at level $i$. The latter is also equivalent to some $P$ with $c(P) = c$ becoming empty after this edge deletion. Let $\langle c_1,\ldots,c_t \rangle$ be the sequence of nodes of $C$ \emph{ordered} according to the time when $v$ has no edge to a node in $P_j$, $1 \leq j \leq t$. We want to compute the probability that $v$ ever joins the cluster centred at $c_j$ during $S(i)$. Note that this event is a consequence of $v$ changing its current cluster center $c_{j'}$ due to all parents in $P(j')$ increasing their level. According to our tie-breaking scheme in $(\ref{eq: breakingTies})$, for this to happen, $c_j$ must be the first among all potential cluster centers $\{c_j,\ldots,c_t\}$ in the random permutation~$\pi$. 
Since $\pi$ is a uniform random permutation, the probability that $c_j$ appears first is $1/(t-j+1)$. By linearity of expectation, the expected number of centers from $C$ whose clusters $v$ joins during $S(i)$ is $\sum_{j=1}^{t} \frac{1}{t-j+1} = O(\log t) = O(\log n)$. This also bounds the number of cluster changes of $v$ during $S(i)$. 
\end{proof}

We next bound the total update time of our decremental algorithm, and also give a bound on the total number of inter-cluster edges during its execution. 

\begin{theorem}\label{thm:decremental decomposition}
There is a decremental algorithm for maintaining a $ (\fract, O (\tfrac{\log{n}}{\fract})) $-decomposition with in expectation at most $ O (\fract n) $ clusters containing non-isolated nodes under a sequence of edge deletions in expected total update time $ O (m \log^3 n / \fract) $ such that, over all deletions, each each becomes an inter-cluster edge at most $ O (\log^2 n) $ times in expectation.
\end{theorem}
\begin{proof}
In a preprocessing step, we first repeat the sampling of the shift values until the maximum shift value is $ \log n/\beta $.
This event happens with probability $1-1/n$ (compare Theorem~\ref{thm: Miller_LDD}) and thus, by the waiting time bound, we need to repeat the sampling only a constant number of times.
Therefore, this preprocessing takes time $ O (n) $, which is subsumed in our claimed bound on the total update time.

We first note that procedure $\Initialize()$ can be implemented in $O (m \log n)$ time. This is because (1) the data-structures $\ell(\cdot), p(\cdot)$ and $c(\cdot)$ are initialized using Algorithm~\ref{alg:ModDijkstra} whose running time is bounded by $O (m \log n)$, (2) for each $u \in V$, the set $P(u)$ can by computed in $O(\deg (u))$ time, which in turn gives that all such sets can be determined in $\sum_{u \in V} O(\deg(u)) = O(m)$ time.

We next analyze the total time over the sequence of all edge deletions. Consider procedure $\Delete(u,v)$ for deletion of an edge $(u,v)$. If edge $(u,v)$ does not lead to a change in the clustering of one of its endpoints, then it can be processed in $O(1)$ time. Otherwise, the end-point whose clustering has changed is inserted into heap $Q$, which can be implemented in $O(\log n)$ time. Now, observe that the computation time spent by procedure $\Delete(u,v)$ is bounded by the number of nodes processed by heap $Q$ after the deletion of edge $(u,v)$, during procedure $\UpdateLevels()$. By construction, the processed nodes are precisely those whose clustering has changed due to the deletion of $(u,v)$, and after the processing, their new clustering its computed. A node $y$ extracted from $Q$ is processed in $O(\deg(y) \log n)$ time, as we will shortly argue. Therefore, we conclude that over the entire sequence of edge deletions, a node $y$ will perform $O(\deg(y) \log n)$ amount of work, each time its clustering changes. By Lemma~\ref{lem: nrClusterChanges}, as long as the level of $y$ is not increased, the clustering of $y$ will change $O(\log n)$ times, in expectation. Since there are at most $\Delta = O(\log n / \beta)$ levels, the expected number of cluster changes for $y$ is bounded by $O((\log^{2} n) / \beta)$. As our analysis applies to any node $y \in V$, we conclude that the expected total update time of our decremental algorithm is
\begin{equation} \label{eqn: runningTime}
	\sum_{y \in V} O \left((\deg(y) \log^3 n) / \beta \right) = O\left((m \log^{3} n / \beta)\right) \, .
\end{equation}

To show our claim that each node $y$ extracted from $Q$ is processed in time $O(\deg(y) \log n)$, we need two observations. First, recall that $P(y)$ can be computed in $O(\deg(y))$ time, and thus the data-structures $\ell(\cdot), p(\cdot)$ and $c(\cdot)$ can be then updated in $O(1)$ time. Second, in the worst-case, $y$ affects the clustering of all its neighbors and inserts them into $Q$. This step can be implemented in $O(\deg(y) \log n)$ time.

We finally show that each each becomes an inter-cluster edge at most $ O (\log^2 n) $ times in expectation.
Fix some arbitrary edge $ e = (x, y) $ and consider the graph $ G $ after an arbitrary number of the adversary's deletions.
We first formulate a necessary condition for $ e $ being an inter-cluster edge and give a bound on the probability of the corresponding event.
Let $ w $ denote the mid-point of $ e $, i.e., the imaginary node in the ``middle'' of edge $ e $ that is at distance $ \tfrac{1}{2} $ to both $ u $ and~$ v $.
Let $ m_{c(x)}(w) $ and $ m_{c(y)}(w) $ denote the shifted distance from $ w $ to $ c (x) $ and $ c (y) $ in $ G $, respectively.
We would like to argue that both $ m_{c(x)}(w) $ and $ m_{c(y)}(w) $ are close to the minimum shifted distance of the mid-point $ w $.
However, we cannot readily apply Lemma~\ref{lem:inter cluster edge implies mid-point condition} as our algorithm does not run on~$ G $; instead it runs on~$ G' $, in which the edge weights are rounded to integers.
However, we can apply Lemma~\ref{lem:inter cluster edge implies mid-point condition} on $ G' $ and get that $ \floor{m_{c(x)}(w)} $ and $ \floor{m_{c(y)}(w)} $ are within $ 1 $ of the minimum rounded shifted distance of the mid-point $ w $.
Thus, $ | \floor{m_{c(x)}(w)} - \floor{m_{c(y)}(w)} | \leq 1 $, which implies that $ | m_{c(x)}(w) - m_{c(y)}(w) | \leq 2 $.
This means that $ | m_{c(x)}(w) - m_{c(y)}(w) | \leq 2 $ is a necessary condition for $ e = (x, y) $ to be an inter-cluster edge.
As the adversary is oblivious to the random choices of our algorithm, we know by Lemma~\ref{lem:probability of mid-point condition} that $ \Pr [| m_{c(x)}(w) - m_{c(y)}(w) | \leq 2 ] \leq 2 \beta $ in each of the graphs created by the adversary's sequence of deletions.

Observe that for each of the endpoints ($ x $ and~$ y $) of $ e $ the level in our decremental algorithm is non-decreasing.
Let $ 0 \leq i \leq \Delta $, and let $ S (i) $, say of length~$ t $, be the (possibly empty) subsequence of edge deletions during which $ \ell (x) = i $.
We show below that the expected number of times that $ e $ becomes an inter-cluster edge during deletions in $ S (i) $ is $ O (\beta \log{n}) $.
It then follows that the total number times $ e $ becomes an inter-cluster edges is $ O (\log^2{n}) $ by linearity of expectation: sum up the number of times $ e $ becomes an inter-cluster edge in each subsequence $ S (i) $ for $ 0 \leq i \leq \Delta $ where $ \Delta \leq O (\log{n} / \beta) $, and repeat the argument for the other endpoint $ y $ of $ e $ as well.

For every $ 1 \leq j \leq t $ define the following events:
\begin{itemize}
\item $ A_j $ is the event that $ e $ becomes an inter-cluster edge after the $j$-th deletion in $ S (i) $, and was not an inter-cluster edge directly before this deletion.
\item $ B_j $ is the event that at least one of the endpoints of $ e $, $ x $ or $ y $, changes its cluster after the $j$-th deletion in $ S (i) $.
\item $ C_j $ is the event that $ e $ is an inter-cluster edge after the $j$-th deletion in $ S (i) $.
\item $ D_j $ is the event that $ | m_{c(x)}(w) - m_{c(y)}(w) | \leq 2 $ after the $j$-th deletion in $ S (i) $, where $ w $ is the mid-point of $ e $.
\end{itemize}

Note that $ e $ can only become an inter-cluster edge if at least one of its endpoints changes its cluster.
Thus, the event~$ A_j $ implies the event $ B_j \wedge C_j $ and therefore $ \Pr [A_j] \leq \Pr [B_j \wedge C_j] $.
Furthermore, by Lemma~\ref{lem:inter cluster edge implies mid-point condition}, the event $ C_j $ implies the event $ D_j $.
We thus have $ \Pr [B_j \wedge C_j] \leq \Pr [B_j \wedge D_j] $.
Observe that the event $ D_j $ only depends on the random choice of the shift values $ \delta $ and that, in the fixed subsequence of deletions $ S (i) $, the event $ B_j $ only depends on the random choice of the permutation~$ \pi $.
Thus, $ B_j $ and $ D_j $ are independent and therefore $ \Pr [B_j \wedge D_j] = \Pr [B_j] \cdot \Pr [D_j] $.
Finally, note that the expected number of indices $ j $ such that the event $ B_j $ happens is at most the expected number of cluster changes for both endpoints of $ e $, as bounded by Lemma~\ref{lem: nrClusterChanges}, and thus $ \sum_{1 \leq i \leq t} \Pr [B_j] = O (\log n) $ for the random permutation~$ \pi $.
It follows that the expected number of times edge $ e $ becomes an inter-cluster edge (i.e., the expected number of indices $ j $ such that event $ A_j $ happens) is
\begin{multline*}
\sum_{1 \leq i \leq t} \Pr [A_j] \leq \sum_{1 \leq i \leq t} \Pr [B_j \wedge C_j] \leq \sum_{1 \leq i \leq t} \Pr [B_j \wedge D_j] = \sum_{1 \leq i \leq t} \Pr [B_j] \cdot \Pr [D_j] \\ \leq \sum_{1 \leq i \leq t} \Pr [B_j] \cdot 2 \beta = 2\beta \cdot \sum_{1 \leq i \leq t} \Pr [B_j] = O (\beta \log{n}) \, ,
\end{multline*}
where the penultimate inequality follows from Lemma~\ref{lem:probability of mid-point condition}.
\end{proof}

Note that in this proof, to bound the total number of inter-cluster edges, we exploited that our two sources of randomness, the random shifts~$ \delta $ and the random permutation~$ \pi $ have different purposes: $ \delta $ influences whether an edge $ e $ is an inter-cluster edge and $ \pi $ influences the number of cluster changes of the endpoints of $ e $.
We have deliberately set up the algorithm in such a way that the independence of the corresponding events can be exploited in the proof.
This is the reason why we explicitly introduced a new random permutation for tie-breaking instead of using the random shifts for this purpose as well.
\begin{remark} \label{remark: LDDrunningTime}
Note that Equation~(\ref{eqn: runningTime}) implies that the total expected update time of Theorem~\ref{thm:decremental decomposition} is $O(m \log^{3} n / \beta)$. For the sake of exposition, we have implemented the ES-tree using a heap, which introduces a $O(\log n)$ factor in the running time. \cite{King99} (Section~2.1.1) gives a faster implementation of the ES-tree that eliminates this extra $O(\log n)$ factor. Thus, using her technique, we can also bring down our running time to $O(m \log^{2} n / \beta)$. This improvement will be particularly useful when applying our dynamic low-diameter decomposition to the construction of dynamic spanners in Section~\ref{sec:spanner}.
\end{remark}

\begin{algorithm2e}[htb!]
\caption{Modified ES-tree}
\label{alg:ES_tree}

\tcp{\textrm{The modified ES-tree is formulated for weighted undirected graphs.}}
\BlankLine

\tcp{\textrm{Internal data structures:
\begin{itemize}
\item $ \pi $: random permutation on $V$
\item $ \delta_v $: random shift of $ v $
\item $ P (v) $: the set of potential parents in the tree
\item $ p (v) $: for every node $ v $ a pointer to its parent in the tree
\item $ c (v) $: for every node $ v $ a pointer to the cluster center
\item $ Q $: global heap whose intended use is to store nodes whose levels might need to be updated
\end{itemize}
}}
\vspace{-3ex}

\BlankLine

\Procedure{\Initialize{}}{
	Initialize using Algorithm~\ref{alg:ModDijkstra}\;
	Set $ \lev (v) $, $ P(v) $, $ p(v) $, $ c(v) $ for every node $ v $ accordingly\;
}

\Procedure{\Delete{$u$, $v$}}{
	\If{$ v \in P (u) $}{
		Remove $ v $ from $ P (u) $\;
		\If{$ P (u) = \emptyset$}{
			Insert $ u $ into heap $ Q $ with key $ \lev (u) $\;
			\UpdateLevels{}\;
		}
	}
}

\Procedure{\UpdateLevels{}}{
	\While{heap $ Q $ is not empty}{\label{line:while loop}
	    Take node $ y $ with minimum key $ \lev (y) $ from heap $ Q $ and remove it from $ Q $ \label{line: take y from Q}\;
		Compute $ P (y) $ as the set of neighbors $ z $ of $ y $ minimizing $ (\lev (z) + w (y, z), \pi(c(z))) $ lexicographically\;
		Set $ p (y) $ as one of the nodes in $ P (y) $\;
		Set $ \lev (y) = \lev (p(y)) + w' (y, p(y) $\;
		\eIf{$ p(y) = s $}{
			Set $ c (y) = y  $\;
		}{
			Set $ c (y) = c( p(y)) $\;
		}

		\ForEach{neighbor $ x $ of $ y $}{\label{line: update neighbors' heaps}
			\If{$ y \in P (x) $}{
				Remove $ y $ from $ P (x) $\;
				\If{$ P (x) = \emptyset$}{
					Insert $ x $ into heap $ Q $ with key $ \lev (x) $ if $ Q $ does not already contain $ x $\;
				}
			}
		}
	}
}
\end{algorithm2e}

\subsection{Fully Dynamic Low-Diameter Decomposition}\label{sec:fully dynamic LDD}

We finally show how to extend the decremental algorithm of Theorem~\ref{thm:decremental decomposition} to a fully dynamic algorithm, allowing also insertions of edges.
\begin{theorem}[Restatement of Theorem~\ref{thm:fully dynamic LDD}]
Given any unweighted, undirected multigraph undergoing edge insertions and deletions, there is a fully dynamic algorithm for maintaining a $ (\fract, O (\tfrac{\log{n}}{\fract})) $-decomposition (with clusters of strong diameter $ O (\tfrac{\log{n}}{\fract}) $ and at most $ \beta m $ inter-cluster edges in expectation) that has expected amortized update time $ O (\log^2{n} / \fract^2) $.
A spanning tree of diameter $ O (\tfrac{\log{n}}{\fract}) $ for each cluster can be maintained in the same time bound.
The expected amortized number of edges to become inter-cluster edges after each update is $ O (\log^2{n} / \fract) $.
These guarantees hold against an oblivious adversary.
\end{theorem}

\begin{proof}
The fully dynamic algorithm proceeds in phases, starting from an empty graph.
For every $ i > 1 $, let $ m_i $ denote the number of edges in the graph at the beginning of phase $ i $.
After $ \fract m_i / 3 $ updates in the graph we end phase $ i $ and start phase $ i+1 $.
At the beginning of each phase we re-initialize the decremental algorithm of Theorem~\ref{thm:decremental decomposition} for maintaining a $ (\fract / 3, 3 \cdot O (\tfrac{\log{n}}{\fract})) $-decomposition.\footnote{Note that for the first constant number of updates this basically amounts to recomputation from scratch at each update.}
Whenever an edge is deleted from the graph, we pass the edge deletion on to the decremental algorithm.
Whenever an edge is inserted to the graph, we do nothing, i.e., we deal with insertions of edges in a completely \emph{lazy} manner.

We first analyze the ratio of inter-cluster edges at any time during phase $ i $.
First observe that the number of inter-cluster edges is at most $ 2 \fract m_i / 3 $ in expectation, where at most $ \fract m_i / 3 $ edges in expectation are contributed by the $ (\fract / 3, 3 \cdot O (\tfrac{\log{n}}{\fract})) $-decomposition of the decremental algorithm and at most $ \fract m_i / 3 $ edges are contributed from inserted edges.
Second, the number of edges in the graph is at least $ m_i - \fract m_i / 3 $, as $ m_i $ is the initial number of edges and at most $ \fract m_i / 3 $ edges have been deleted.
Thus, the ratio of inter-cluster edges is at most
\begin{equation*}
\frac{2 \fract m_i / 3}{m_i - \fract m_i/3} = \frac{2 \fract}{3 - \fract} \leq \frac{2 \fract}{2 + \fract - \fract} = \fract \, .
\end{equation*}
Our fully dynamic algorithm therefore correctly maintains a $ (\fract, O (\tfrac{\log{n}}{\fract})) $-decomposition.

We now analyze the amortized update time of the algorithm.
Start with an empty graph and consider a sequence of $ q $ updates.
Let $ k $ denote the number of the phase after the $ q $-th update.
Then $ q $ can be written as $ q = \sum_{1 \leq i < k} \fract m_i / 3 + t $, where $ t $ is the number of updates in phase $ k $.
For every phase $ i $ that has been started, we spend time $ O (m_i \log^2 n / \fract) $ by Theorem~\ref{thm:decremental decomposition} and Remark~\ref{remark: LDDrunningTime}. We know that $ t \leq \fract m_k / 3 $ and in particular we also have $ m_k \leq \sum_{1 \leq i \leq k-1} \fract m_i / 3 $ as every edge that is contained in the graph at the beginning of phase $ k $ has been inserted in one of the previous phases.
We can thus bound the amortized spent by the algorithm for $ q $ updates by
\begin{multline*}
\frac{\sum_{1 \leq i \leq k-1} O (m_i \log^2 n / \fract) + O (m_k \log^2 n / \fract)}{\sum_{1 \leq i \leq k-1} \fract m_i / 3} \\
 \leq \frac{\sum_{1 \leq i \leq k-1} O (m_i \log^2 n / \fract) + O (\sum_{1 \leq i \leq k-1} m_i \log^2 n)}{\sum_{1 \leq i \leq k-1} \fract m_i / 3} = O \left( \frac{\log^2 n}{\fract^2} \right) \, .
\end{multline*}

Finally, we analyze the amortized number of edges to become inter-cluster edges per update.
For every phase $ i $ that has been started, we have a total number of $ O (m_i \log^2 n) $ edges that become inter-cluster edges in the decremental algorithm by Theorem~\ref{thm:decremental decomposition}.
Additionally, at most $ \fract m_i / 3 = O (m_i) $ inserted edges could also become inter-cluster edges.
We can thus bound the amortized number of edges to become inter-cluster per update by
\begin{multline*}
\frac{\sum_{1 \leq i \leq k-1} O (m_i \log^2 n) + O (m_k \log^2 n)}{\sum_{1 \leq i \leq k-1} \fract m_i / 3} \\
 \leq \frac{\sum_{1 \leq i \leq k-1} O (m_i \log^2 n) + O (\sum_{1 \leq i \leq k-1} \fract m_i \log^2 n)}{\sum_{1 \leq i \leq k-1} \fract m_i / 3} = O \left( \frac{\log^2 n}{\fract} \right) \, .
\end{multline*}
\end{proof}

%% file: static.tex
\subsection{Static Low-Diameter Decomposition}~\label{sec:staticLDD}

In the following, we review the static algorithm for constructing a low-diameter decomposition clustering due to Miller et al.~\cite{MillerPX13}. Let $G=(V,E)$ be an unweighted, undirected multigraph $G$, and let $\beta \in (0,1)$ be some parameter. Our goal is to assign each node $u$ to exactly one node $c(u)$ from $V$. Let $C(u) \subset V$ denote the set of nodes assigned to node $u$, i.e., $C(u) \coloneqq \set{v \in V}{ c(v) = u}$. For each node $u$, we initially set $C(u)= \varnothing $ and pick independently a \emph{shift} value $\delta_u$ from $\expdis(\beta)$. Next, we assign each node $u$ to a node $v$, i.e., set $c(u) = v$ and add $u$ to $C(v)$, if $v$ is the node that minimizes the \emph{shifted distance} $m_v(u)\coloneqq \dist(u,v) - \delta_v$. Finally, we output all clusters that are non-empty. The above procedure is summarized in Algorithm~\ref{alg:partition}.

\begin{algorithm2e}[htb!]
\caption{Partitioning Using Exponentially Shifted Shortest Paths}
\label{alg:partition}

\BlankLine
\KwData{Multigraph $G=(V,E)$, parameter $\beta \in (0,1)$}
\KwResult{Decomposition of $G$}
\BlankLine

For each $u \in V$, set $C(u)=\varnothing$ and pick $\delta_u$ independently from $\expdis(\beta)$ \;
Assign each $u \in V$ to $c(u) \gets \arg \min_{v \in V} \{\dist(u,v) - \delta_v\}$\;
For each $v \in V$, set $C(u) \gets \set{v \in V}{c(v) = u}$\;
Return the clustering $\set{C(u)}{C(u) \neq \varnothing}$
\end{algorithm2e}

The following theorem gives bounds on the strong diameter and the number of inter-cluster edges output by the above partitioning.

\begin{theorem}[\cite{MillerPX13}, Theorem~1.2] \label{thm: Miller_LDD} Given an undirected, unweighted multigraph graph $G=(V,E)$ and a parameter $\beta \in (0,1)$, Algorithm~\ref{alg:partition} produces a $(\beta, 2 d \cdot(\log n/\beta))$-decomposition such that the guarantee on the number of inter-cluster edges holds in expectation, while the diameter bound holds with probability at least $1-1/n^d$, for any $d \geq 1$.
\end{theorem}

Here, the the diameter bound holds when the maximum shift value of any node is at most $ d \log n/\beta $, which happens with probability $1-1/n^d$.
We remark that in the work of Miller et al., the above guarantees are stated only for undirected, unweighted \emph{simple} graphs. However, by Lemma~4.4 in~\cite{MillerPX13}, we get that each edge $e \in E$~(regardless of whether $E$ allows parallel edges) is an inter-cluster edges with probability at most $\beta$. By linearity of expectation, it follows that the (expected) number of inter-cluster edges in the resulting decomposition is at most $\beta m$, thus showing that the algorithm naturally extends to multigraphs.

For technical reasons, it is not sufficient in the analysis of our decremental LDD algorithm to apply Theorem~\ref{thm: Miller_LDD} in a black-box manner.
We thus review the crucial properties of the clustering algorithm, which we will exploit for bounding the number of changes made to inter-cluster edges in the decremental algorithm. Following~\cite{MillerPX13}, for each edge $e=(u,v) \in E$, let $w$ be the \emph{mid-point} of $e$, i.e., the imaginary node in the ``middle'' of edge $ e $ that is at distance $ \tfrac{1}{2} $ to both $ u $ and $ v $. Lemma~4.3 in \cite{MillerPX13} states that if $u$ and $v$ belong to two different clusters, i.e., $c(u) \neq c(v)$, then the shifted-distances $m_{c(u)}(w)$ and $m_{c(v)}(w)$ are within $1$ of the minimum shifted distance to $w$.

\begin{lemma}[\cite{MillerPX13}]\label{lem:inter cluster edge implies mid-point condition}
Let $e=(u,v)$ be an edge with mid-point $w$ such that $c(u) \neq c(v)$ in Algorithm~\ref{alg:partition}. Then $m_{c(u)}(w)$ and $m_{c(v)}(w)$ are within $1$ of the minimum shifted distance to $w$. 
\end{lemma}

Lemma 4.4 of~\cite{MillerPX13} shows that the probability that the smallest and the second smallest shifted distances to $w$ are within $c$ of each other is at most $c \cdot \beta$.

\begin{lemma}[\cite{MillerPX13}]\label{lem:probability of mid-point condition}
Let $ e = (u,v) $ be an edge with mid-point $ w $.
Then \[ \Pr [| m_{c(u)}(w) - m_{c(v)}(w) | \leq c] \leq c \cdot \beta. \]
\end{lemma}

Setting $c=1$, this gives the desired bound of $ \beta $ for the probability of an edge being an inter-cluster edge in Theorem~\ref{thm: Miller_LDD}.

\paragraph*{Implementation.} 
Na\"ively, we could implement Algorithm~\ref{alg:partition} by computing $c(u)$ for each node $u \in V$ in $\tilde{O}(m)$, thus leading to a $\tilde{O}(mn)$ time algorithm. In the following, using standard techniques, we show that this running time can be reduced to $\tilde{O}(m)$.

To this end, let $\delta_{\max} \coloneqq \max_{u \in V} \delta_u$. We begin with the following augmentation of the input graph $G$: add a new source $s$ to $G$ and edges $(s,u)$ of weight $(\delta_{\max} - \delta_u) \geq 0$, for every $u \in V$. Let $\hat{G}=(V \cup \{s\},\hat{E},\hat{w})$ denote the resulting graph. We claim that the sub-trees below the source $s$ in the shortest path tree of $\hat{G}$ rooted at $s$ give us the clustering output by Algorithm~\ref{alg:partition} for the graph $G$. To see this, suppose that we instead added edges of weight $-\delta_u$ to $s$, for every $u \in V$. Then it is easy to check that for every $u \in V$, distance between $s$ and $u$ is exactly $ \min_{v \in V} (\dist(u,v) - \delta_v) = \min_{v \in V} m_v(u)$. Thus the node $v$ attaining the minimum is exactly the root of the sub-tree below the source $s$ that contains $v$. Now, adding $\delta_{\max}$ to all edges incident to the source increases all distances to $s$ by $\delta_{\max}$, and thus does not affect the shortest path tree.

Now, note that we could use Dijkstra's algorithm to construct the shortest path tree of $\hat{G}$, and modify it appropriately to output the clustering. However, for reasons that will become clear in the next section, we need to modify Dijkstra's algorithm in a specific way. This modification can be viewed as mimicking a BFS computation on a graph with special \emph{integral} edge lengths. 

We start by observing that due to the random shift values, the weight of the edges incident to the source $s$ in $\hat{G}$ are not integers. Since we only want to deal with \emph{integral} weights, we round down all the $\delta_u$ values to $\floor{\delta_u}$ and modify the weights of these edges using the new rounded values. Let $G'=(V \cup \{s\}, \hat{E}, w')$ denote the modified graph. Note that due to the rounding, we need to introduce some tie-breaking scheme in $G'$, such that  every clustering of $G'$ matches exactly the same clustering in $\hat{G}$, and vice versa. Naturally, the fractional parts of the rounded values, i.e., $\delta_u - \floor{\delta_u}$, define an ordering on the nodes (if they are sorted in ascending order), and this ordering can be in turn used to break ties whenever two rounded distances are equal in $G'$. In their PRAM implementation, Miller et al.~\cite{MillerPX13} observed that this ordering can emulated by a random permutation. This is due to the fact that the shifts are generated independently, and that the exponential distribution is memoryless. 

The main motivation for using random permutations in previous works was to avoid errors that might arise from the machine precision. In our work, breaking ties according to a random permutation on the nodes is one of their algorithmic ingredients that allows us to obtain an efficient dynamic variant of the clustering. Below, we give specific implementation details about how our clustering interacts with random priorities in the static setting. 

Given the graph $G'$ and a distinguished source node $s \in V'$,  Dijkstra's classical algorithm maintains an upper-bound on the shortest-path distance between each node $u \in V$ and $s$, denoted by $\ell(u)$. Initially, it sets $\ell(u) = \infty$, for each $u \in V$ and $\ell(s) = 0$. It also marks every node unvisited. Moreover, for each node $u \in V$, the algorithm also maintains a pointer $p(u)$ (initially set to $\nil$), which denotes the parent of $u$ in the current tree rooted at $s$. Using these pointers, we can maintain the cluster pointer $c(u)$, for each $u \in V$. This follows from the observation that in order to compute the cluster of $u$, it suffices to know the cluster of its parent. Formally we have the following rule.

\begin{observation} \label{obs: clusterUpdate}
Let $p(\cdot)$ be the parent pointers. Then for each $u \in V$, we can determine the cluster pointer $c(u)$ using the following rule:

\[
	c(u) =  \begin{cases} 
    u & \text{\emph{if} } p(u)=s \\
    c(p(u))      & \text{\emph{otherwise}}.
  \end{cases}
\]
\end{observation}

Now, at each iteration, Dijkstra's algorithm selects an unvisited node $u$ with the smallest $\ell(u)$, marks it as visited, and \emph{relaxes} all its edges. In the standard relaxation, for each edge $(u,v) \in E'$ the algorithm sets $\ell(v) \gets \min\{\ell(v), \ell(u) + w'(u,v)\}$ and updates $p(v)$ accordingly. Here, we present a relaxation according to the following tie-breaking scheme. Let $\pi$ be a random permutation on $V$. For $u, v \in V$, we write $\pi(u) < \pi(v)$ if $u$ appears before $v$ in the permutation $\pi$.
Now, when relaxing an edge $(u,v) \in E'$, we set $u$ to be the parent of $v$, i.e., $p(v) = u$, and $\ell(v) = \ell(u) + w'(u,v)$, if the following holds
\begin{align} \label{eq: breakingTies}
\begin{split}
\ell(v) & > \ell(u) + w'(u,v), \text{ or} \\
\ell(v) & = \ell(u) + w'(u,v) \text{ and } \pi(c(v)) > \pi(c(u)).
\end{split}
\end{align}

After each edge relaxation, we also update the cluster pointers using  Observation~\ref{obs: clusterUpdate}. We continue the algorithm until every node is visited. As usual, we maintain the unvisited nodes in a heap $Q$, keyed by the their estimates $\ell(v)$. This procedure is summarized in Algorithm~\ref{alg:ModDijkstra}.

\begin{algorithm2e}[htb!]
\caption{Modified Dijkstra}
\label{alg:ModDijkstra}
\KwData{Graph $G'=(V \cup \{s\},E',w')$}
\KwResult{Decomposition of $G$}
\BlankLine

Generate random permutation $\pi$ on $V$\;
\ForEach{$u \in V$}{
Set $\ell(u) = \infty$\;
Set $p(u) = \nil$\;
Set $c(u) = \nil$\;
}

Set $\ell(s) = 0$\;
\BlankLine

Add every $u \in V \cup \{s\}$ into heap $Q$ with key $\ell(u)$\;
\While{heap $ Q $ is not empty}{
	    Take node $ u $ with minimum key $ \lev (u) $ from heap $ Q $ and remove it from $ Q $ \;
		\ForEach{neighbor $ v $ of $ u $}{
		\Relax{$u$, $v$, $w'$, $\fraction$} \;
		\eIf{$ p(v) = s $}{
			Set $ c (v) = v  $\;
		}{
			Set $ c (v) = c( p(v)) $\;
		}
	}
}

\Procedure{\Relax{$u$, $v$, $w'$, $\fraction$}}{
	\uIf{$\ell(v) > \ell(u) + w'(u,v)$}{
		Set $\ell(v) = \ell(u) + w'(u,v)$\;
		Set $p(v) = u$\;
		}
	\uElseIf{$\ell(v) = \ell(u) + w'(u,v)$ \emph{and} $\pi(c(v)) > \pi(c(u))$}{
     Set $p(v) = u$ \;}
}
\end{algorithm2e}

Correctness of Algorithm~\ref{alg:ModDijkstra} follows by our above discussion. Moreover, the running time of the algorithm is asymptotically bounded by the running time of Dijkstra's classical algorithm and the time to generate a random permutation. It is well known that the former runs in $\tilde{O}(m)$ time and the latter can be generated in $O(n)$ time (see e.g., Knuth Shuffle~\cite{BermanK76}), thus giving us a total $\tilde{O}(m)$ time.

%% file: spanner.tex
\section{Dynamic Spanner Algorithm}\label{sec:spanner}

\subsection{Static Spanner Construction}
In the following we review and adapt the static algorithm for constructing sparse low-stretch spanners due to Elkin and Neiman~\cite{ElkinN17}. Let $G=(V,E)$ be an unweighted, undirected graph on $n$ nodes, and let $k \geq 1$ be an integer. For every $u \in V$, we denote by $N(u)$ the set of all nodes incident to $u$. Recall that $\expdis(\beta)$ denotes the exponential distribution with parameter $\beta$. In what follows, we set $\beta = \log (cn)/k$, where $c>3$ denotes the success probability. A $2 k - 1$-\emph{spanner} of $G$ is a a subgraph $H=(V,E')$ such that for every $u,v \in V$, $\dist_H(u,v) \leq 2 k - 1 \cdot \dist_G(u,v)$. We refer to $2 k - 1$ and $|E'|$ as the \emph{stretch} and \emph{size} of $H$, respectively. 

We next review some useful notation. Let $\delta_u$ be the shift value of node $u \in V$. For each $x, u \in V$, recall that $m_u(x) = \dist_G(x,u) - \delta_u$ is the shifted distance of $x$ with respect to $u$, and let $p_u(x)$ denote the neighbor of $x$ that lies on a shortest path from $x$ to $u$. Also, for every node $x \in V$, let $m(x) = \min_{u \in V} \{ m_u(x) \}$ be the minimum shifted distance. Using our clustering notation from Section~\ref{sec:LDD}, it follows that $c(x) = \arg \min_{u \in V} \{ m_u(x) \}$, and thus $m(x) = m_{c(x)}(x)$.

We now present an algorithm that constructs spanners using exponential random-shift clustering. Specifically, we initially set $H=(V,\emptyset)$, and for each node $u \in V$, we independently pick a shift value $\delta_u$ from $\expdis(\beta)$. Then, for every $x \in V$, we add to the spanner $H$ the following set of edges
\begin{equation}
C(x) = \set{(x, p_u(x))}{m_u(x) \leq m(x) + 1} \, . \label{eqn: SpannerEqn}
\end{equation}

The following theorem give bounds on the stretch and the size of the spanner output by the above algorithm.

\begin{theorem}[\cite{ElkinN17}] For any unweighted, undirected simple graph $G=(V,E)$ on $n$ nodes, any integer $k \geq 1$, $c \geq 3$, there is a randomized algorithm that with probability at least $ 1 - \tfrac{2}{c} $ computes a spanner $H$ with stretch $2k-1$ and size at most $ (cn)^{1+1/k} $.
\end{theorem}

Our analysis will rely on the following useful properties of the above algorithm.

\begin{claim}[\cite{ElkinN17}]
The expected size of $ H $ is at most $ (c n)^{1/k} \cdot n $.
\end{claim}

\begin{claim}[\cite{ElkinN17}]\label{claim:bound on random shifts}
With probability at least $1-1/c$, it holds that $\delta_u < k$ for all $u \in V$.
\end{claim}

\begin{claim}[\cite{ElkinN17}]\label{claim:bound on depth of tree} Assume $\delta_u < k$ for all $u \in V$. Then for any $x \in V$, if $u$ is the node minimizing $m_u(x)$, i.e., $u = c(x)$, then $\dist_G(u,x) < k$.
\end{claim}

As argued by Elkin and Neiman, Claim~\ref{claim:bound on depth of tree} implies that the stretch of the spanner is at most $ 2k - 1 $.
Thus, the reason reason why the stretch guarantee is probabilistic is Claim~\ref{claim:bound on random shifts}.

\paragraph*{Implementation.} In the description of the spanner construction, it is not clear how to compute in nearly-linear time the set of edges $C(x)$ in Equation~(\ref{eqn: SpannerEqn}), for every node $x \in V$. To address this, we give an equivalent definition of $C(x)$, which better decouples the properties that the edges belonging to this set satisfy. Specifically, we define the set of edges
\begin{equation}
	C'(x) = \set{(x,y)}{y \in N(x) \text{ and } m_{c(y)}(x) \leq m(x) + 1} \, ,\label{eqn: equivalentSpannerEqn}
\end{equation}
and then show that $C(x) = C'(x)$. 

To this end, we will show that (a) $C(x) \subseteq C'(x)$ and (b) $C'(x) \subseteq C(x)$. Let $(x,y) \in C(x)$, where $y = p_u(x)$. By definition of $p_u(x)$, we have that $y \in N(x)$. We next show that $m_{c(y)}(x) \leq m(x) + 1$, which in turn proves (a). Indeed, 
\[
	m_{c(y)}(x) = m_{c(y)}(y) + 1 = m(y) + 1 \leq m_u(y) + 1 = m_u(x) \leq m(x)+1,
\]
where the last inequality follows from Equation~(\ref{eqn: SpannerEqn}). For showing the other containment, i.e., proving (b), let $(x,y) \in C'(x)$. Then we need to prove that there exists some $u \in V$ such that $y=p_u(x)$ and $m_u(x) \leq m(x) +1$. This follows by simply setting $u=c(y)$ and using Equation~(\ref{eqn: equivalentSpannerEqn}).

Now, similarly to the static low-diameter decomposition in Section~\ref{sec:staticLDD}, we augment the input graph~$G$ by adding a new source $s$ to $G$ and edges $(s,x)$ of weight $(\delta_{max} - \delta_x) \geq 0$, for every $x \in V$, where $\delta_{\max} = \max_{x \in V} \{\delta_x \}$. Recall that in the resulting graph $\hat{G}=(V \cup \{s\},\hat{E},\hat{w})$, for every $x \in V$, the node $u$ attaining the minimum $m(x)$ is exactly the root of the sub-tree below the source $s$ that contains $u$. Thus, we could use Dijkstra's algorithm to construct the shortest path tree of $\hat{G}$, and augment it appropriately to output the edge sets $C'(x)$, which in turn give us the spanner $H$.

However, in the dynamic setting, it is crucial for our algorithm to deal only with integral edge weights. To address this, we round down all the $\delta_u$ values to $\floor{\delta_u}$ and modify the weights of the edges incident to the source $s$ in $\hat{G}$. Let $G'=(V \cup \{s\}, E',w')$ be the resulting graph, and let $\floor{m_u(x)}$ denote the rounded shifted distances. Whenever two rounded distances are the same, we break ties using the permutation $ \pi $ on the nodes induced by the fractional values of the random shift values. Thus, the edge set $C'(x)$ is given by
\begin{equation*} 
	C'(x) = \set{(x,y)}{y \in N(x),~\floor{m_{c(y)(x)}} \leq \floor{m(x)} + 1 \text{ and } \pi(c(y)) < \pi(c(x))}.
\end{equation*}

Finally, we observe that the definition of the above set can be further simplified by using the facts that $m_{c(y)}(x) = m_{c(y)}(y) + 1$ and $\floor{m_{c(y)}(x)} \geq \floor{m(x)}$, that is
\begin{equation}
   C'(x) = \set{(x,y)}{y \in N(x),~\floor{m(y)} = \floor{m(x)} - 1 \text{ or } [\floor{m(y)} = \floor{m(x)} \text{ and } \pi(c(y)) < \pi(c(x))]} \, . \label{eqn: roundedSpannerEqn}	
\end{equation}

Interpreting the above set in terms of the shortest-path tree output by Dijkstra's algorithm, we get that for any $x \in V$, we add the edge $(x,y)$ to the spanner $H$, if $y$ is a neighbor one level above the level or $x$, or if $x$ and $y$ are at the same level, and the cluster $y$ belongs to appears before in the permutation when compared to the cluster $x$ belongs to.
By Claim~\ref{claim:bound on depth of tree} the shortest-path tree has depth at most $ 2 k $ with high probability.

Now observe that the randomized properties of this spanner construction only depend on the integer parts of the random shift values and the permutation~$ \pi $ on the nodes induced by the order statistics of the fractional parts of the random shift values.
Similar to the argument of Miller et al.~\cite{MillerPX13} for low-diameter decompositions, it can be argued that due to memorylessness of the exponential distribution, one might as well use a uniformly sampled random permutation~$ \pi $ instead to obtain a spanner with the same probabilistic properties.

\subsection{Dynamic Spanner Algorithm}

Spanners have a useful property called \emph{decomposability}: Assume we are given a graph $ G = (V, E) $ with a partition into two subgraphs $ G_1 = (V, E_1) $ and $ G_2 = (V, E_2) $.
If $ H_1 = (V, F_1) $ is a spanner of $ G_1 $ and $ H_2 = (V, F_2) $ is a spanner of $ G_2 $, both of stretch $ t $, then $ H = (V, F_1 \cup F_2) $ is a spanner of $ G $.
This property allows for a reduction that turns decremental algorithms into fully dynamic ones at the expense of logarithmic overhead in size and update time, as it has been observed by Baswana et al.~\cite{BaswanaKS12}.
\begin{lemma}[Implicit in~\cite{BaswanaKS12}]\label{lem:spanner decremental to fully dynamic}
If there is a decremental algorithm for maintaining a spanner of stretch $ t $ and expected size $ s (n) $ with total update time $ m \cdot u (m, n) $, then there is a fully dynamic algorithm for maintaining a spanner of stretch $ t $ and expected size $ s (n) \cdot O (\log{n}) $ with amortized update time $ u (m, n) \cdot O (\log{n}) $.
\end{lemma}

In the remainder of this section, we explain how the techniques we developed in Section~\ref{sec:LDD} allow for a decremental implementation of the spanner construction explained above.
\begin{theorem}
Given any unweighted, undirected graph undergoing edge deletions, there is a decremental algorithm for maintaining a spanner of stretch $ 2k - 1 $ and expected size $ O (n^{1 + 1/k}) $ that has expected total update time $ O (k m \log{n}) $.
These guarantees hold against an oblivious adversary.
\end{theorem}
Using the reduction of Lemma~\ref{lem:spanner decremental to fully dynamic}, these guarantees carry over to the fully dynamic setting.
\begin{theorem}[Restatement of Theorem~\ref{thm:fully dynamic spanner}]
Given any unweighted, undirected graph undergoing edge insertions and deletions, there is a fully dynamic algorithm for maintaining a spanner of stretch $ 2k - 1 $ and expected size $ O (n^{1 + 1/k} \log{n}) $ that has expected amortized update time $ O (k \log^2{n}) $.
These guarantees hold against an oblivious adversary.
\end{theorem}

The decremental algorithm is obtained as follows:
In a preprocessing step, the algorithm samples the random shift values for the nodes from the exponential distribution and additionally a uniformly random permutation $ \pi $ on the nodes.
The sampling of the random shift values is repeated until $ \delta_u < k $ for all $u \in V$.
By Claim~\ref{claim:bound on random shifts} this condition holds with probability at least $ 1 - 1/c $.
Thus, by the waiting time bound, we need to repeat the sampling at most a constant number of times for the condition to hold.
As each round of sampling takes time $ O (n) $, this preprocessing step requires an additional $ O (n) $ in the total update time.

We can then readily use Algorithm~\ref{alg:ES_tree} from Section~\ref{sec:decremental LDD} to maintain a shortest path tree up to depth $ 2 k $ from $ s $ in the graph $ G' $, as defined above.
For maintaining the spanner dynamically, we need to extend the algorithm to maintain the set $ C' (x) $ for every node $ x $.
Using the arguments introduced in Section~\ref{sec:decremental LDD}, this can be done in a straightforward way:
Every time a node $ x $ changes its level in the tree or changes its cluster $ c (x) $, it (1) recomputes the set $ C' (x) $ in time $ O (\deg (x)) $ and stores it in a hash set and (2) informs each neighbor about the change and updates the set $ C' (y) $ of each neighbor $ y $ by setting the entry corresponding to the edge $ (x, y) $ accordingly.
Both (1) and (2) require (expected) time $ O (\deg (x)) $.
As the maximum level in the tree is $ O (k) $ and at each node changes its clustering at a fixed level at most $ O (\log n) $ times in expectation, the expected total update time of our algorithm is $ O (k m \log{n}) $ as desired.